\documentclass[journal]{IEEEtran}
\usepackage{amsthm,amsmath}
\usepackage{booktabs}
\usepackage{pgfplots}
\usepackage{graphicx}
\usepackage{amsfonts}
\usepackage{amssymb}
\usepackage{cool}
\usepackage{mathrsfs}
\usepackage{tikz}
\usepackage{cite}
\usetikzlibrary{shapes,arrows,backgrounds}
\usetikzlibrary{shadows,calc,fit,arrows}
\usetikzlibrary{patterns}
\usetikzlibrary{arrows.meta, calc, decorations.markings, quotes}
\usepgfplotslibrary{colormaps} 
\usetikzlibrary{pgfplots.colormaps} 
\usetikzlibrary[pgfplots.colormaps] 
\usetikzlibrary{shapes.geometric}
\usepackage{adjustbox}
\usepackage{array, multirow}
\newcolumntype{C}[1]{>{\centering\arraybackslash}p{#1}}
\usetikzlibrary{shapes}
\usetikzlibrary{decorations.markings}
\usepgfplotslibrary{groupplots,fillbetween}
\pgfplotsset{compat=1.15}
\newtheorem{theorem}{Theorem}

\newtheorem{lemma}{Lemma}

\theoremstyle{remark}
\newtheorem{remark}{Remark}
\newtheorem{example}{Example}
\newtheorem{definition}{Definition}

\allowdisplaybreaks

\begin{document}

\title{Capacity Region Bounds for the $K-$user Dispersive Nonlinear Optical WDM Channel with Peak Power Constraints}%

\author{
\IEEEauthorblockN{Viswanathan~Ramachandran, Gabriele~Liga, \textit{Member, IEEE}, Astrid~Barreiro, \textit{Student Member, IEEE}, and Alex~Alvarado, \textit{Senior Member, IEEE}}\thanks{The authors are with the Information and Communication Theory Lab, Signal Processing
Systems (SPS) Group, Department of Electrical Engineering, Eindhoven
University of Technology, 5600 MB Eindhoven, The Netherlands (e-mails:
\{v.ramachandran,g.liga, a.barreiro.berrio,a.alvarado\}@tue.nl). 

The work of V.  Ramachandran, A. Barreiro and A. Alvarado has received funding from the European Research Council (ERC) under the European Union's Horizon 2020 research and innovation programme (grant agreement No 757791). The work of G.~Liga is funded by the EuroTechPostdoc programme under the European Union's Horizon 2020 research and innovation programme (Marie Sk\l{}odowska-Curie grant agreement No 754462).

This work has appeared in part at the 2022 Optical Fiber Communications Conference (OFC)~\cite{viswanathan2021wdm}.
}
}
\maketitle

\begin{abstract}
    It is known that fiber nonlinearities induce crosstalk in a wavelength division multiplexed (WDM) system, which limits the capacity of such systems as the transmitted signal power is increased. A network user in a WDM system is an entity that operates around a given optical wavelength. Traditionally, the channel capacity of a WDM system has been analyzed under different assumptions for the transmitted signals of the other users, while treating the interference arising from these users as noise. In this paper, we instead take a multi-user information theoretic view and treat the optical WDM system impaired by cross-phase modulation and dispersion as an interference channel. { We characterize an outer bound on the capacity region of simultaneously achievable rate pairs, assuming a simplified $K$-user perturbative channel model using genie-aided techniques.} Furthermore, an achievable rate region is obtained by time-sharing between certain single-user strategies. It is shown that such time-sharing can achieve better rate tuples compared to treating nonlinear interference as noise. 
    For the single-polarization single-span system under consideration and a power $4.4$~dB above the optimum launch power, treating nonlinear interference as noise results in a rate of $1.67$~bit/sym, while time-sharing gives a rate of $6.33$~bit/sym.
\end{abstract}

\section{Introduction}
In a point-to-point wavelength division multiplexing (WDM) system, independent data from different users across different wavelengths are multiplexed into a single optical fiber using several optical transmitters, with corresponding demultiplexing at the receiver side. The nonlinear {Kerr effect} in an optical fiber causes the signal in one wavelength to interfere with the signals in other wavelengths. 
The combination of nonlinear effects with chromatic dispersion (group velocity dispersion) and noise, result in  a stochastic nonlinear channel with memory. Such a channel is described by the (noisy) nonlinear Schr{\"o}dinger equation (NLSE) (or the Manakov equation in case of dual polarization systems), which considers intra-channel effects like self-phase modulation (SPM), and inter-channel effects such as cross-phase modulation (XPM) and four-wave mixing (FWM). SPM can be compensated for using digital backpropagation (DBP)~\cite{ip2008compensation}, while FWM is known to be negligible compared to XPM for most practical systems.
Full XPM compensation, on the other hand, would entail joint detection of multiple channels that is prohibitively complex. As a result, XPM represents the dominant transmission bottleneck in WDM systems. This paper focuses on XPM-dominated systems.

\begin{figure*}
\begin{center}
\scalebox{1.0}{\begin{tikzpicture}[thick]
\node (e1) at (-2.9,1.5) [rectangle, draw, minimum height=1.0cm]{\scriptsize ${E_1}$};
\node (e11) at (-2.9,0.3) [rectangle, draw, minimum height=1.0cm]{\scriptsize ${E_2}$};
\node (e2) at (-2.9,-1.5) [rectangle, draw, minimum height=1.0cm]{\scriptsize ${E_K}$};
\foreach \i in {-0.5,-0.6,-0.7} {\fill (-2.85,\i) circle(0.02cm);}
\foreach \i in {-0.5,-0.6,-0.7} {\fill (-0.98,\i) circle(0.02cm);}
\node (e3) at (-1.02,1.5) [rectangle, draw, minimum height=1.0cm,minimum width=0.28cm,align=center]{\scriptsize {E-O}}; 
\node (e33) at (-1.02,0.3) [rectangle, draw, minimum height=1.0cm,align=center]{\scriptsize {E-O}}; 
\node (e4) at (-1.02,-1.5) [rectangle, draw, minimum height=1.0cm,align=center]{\scriptsize {E-O}};
\node (e5) at (7.5,1.5) [rectangle, draw, minimum height=1.0cm,align=center]{\scriptsize {MF}};
\node (e55) at (7.5,0.3) [rectangle, draw, minimum height=1.0cm,align=center]{\scriptsize {MF}};
\node (e6) at (7.5,-1.5) [rectangle, draw, minimum height=1.0cm,align=center]{\scriptsize {MF}};
\node (e7) at (5.45,1.5) [rectangle, draw, minimum height=1.0cm,minimum width=0.05cm,align=center]{\scriptsize {O-E}};
\node (e77) at (5.45,0.3) [rectangle, draw, minimum height=1.0cm,minimum width=0.05cm,align=center]{\scriptsize {O-E}};
\node (e8) at (5.45,-1.5) [rectangle, draw, minimum height=1.0cm,minimum width=0.05cm,align=center]{\scriptsize {O-E}};
\node (e9) at (6.45,1.5) [rectangle, draw, minimum height=1.0cm,align=center]{\scriptsize {DBP}};
\node (e99) at (6.45,0.3) [rectangle, draw, minimum height=1.0cm,align=center]{\scriptsize {DBP}};
\node (e10) at (6.45,-1.5) [rectangle, draw, minimum height=1.0cm,align=center]{\scriptsize {DBP}};
\node[trapezium,draw,rotate=-90,minimum height=0.7cm,minimum width=0.4cm] (t1) at (0.1,0) {\small {WDM MUX}};
\node[trapezium,draw,rotate=90,minimum height=0.68cm,minimum width=0.25cm] (t2) at (4.40,0) {\small {WDM DE-MUX}};
\node (c1) at (1.6,0.1) [circle, draw, minimum height=0.5cm,thick]{};
\node (c2) at (1.5,0.1) [circle, draw, minimum height=0.5cm,thick]{};
\node[isosceles triangle, draw, minimum size =.9cm] (T)at (2.4,-0.25){};
\node (d1) at (9.3,1.5) [rectangle, draw, right, minimum height=1.0cm]{\scriptsize ${D_1}$};
\node (d11) at (9.3,0.3) [rectangle, draw, right, minimum height=1.0cm]{\scriptsize ${D_2}$};
\node (d2) at (9.3,-1.5) [rectangle, draw, right, minimum height=1.0cm]{\scriptsize ${D_K}$};
\foreach \i in {-0.5,-0.6,-0.7} {\fill (9.65,\i) circle(0.02cm);}
\foreach \i in {-0.5,-0.6,-0.7} {\fill (5.45,\i) circle(0.02cm);}
\foreach \i in {-0.5,-0.6,-0.7} {\fill (6.45,\i) circle(0.02cm);}
\foreach \i in {-0.5,-0.6,-0.7} {\fill (7.4,\i) circle(0.02cm);}
\node () at (1.57,-0.57) {\small {SSMF}};
\node () at (2.55,-0.25) {\footnotesize {EDFA}};
\node () at (1.02,-1.17) {\small $A(t,0)$};
\node () at (3.4,-1.17) {\small $A(t,L)$};
\node () at (-2.21,1.8) {\small $X_1^n$};
\node () at (-2.21,0.55) {\small $X_2^n$};
\node () at (-2.16,-1.8) {\small $X_K^n$};
\node () at (9,1.8) {\small $Y_1^n$};
\node () at (9,0.55) {\small $Y_2^n$};
\node () at (9,-1.8) {\small $Y_K^n$};
\node () at (2.2,2.5) {{Interference Channel \eqref{eq:Kuserapprox}}};
\node (distr) at (2.2,3) {$p(Y_1^n,\cdots,Y_K^n|X_1^n,\cdots,X_K^n)$};
\node[draw, ellipse, minimum width=.4cm, minimum height=4cm] at (-1.75,0) (ell1) {};
\node[draw, ellipse, minimum width=.4cm, minimum height=4cm] at (8.55,0) (ell2) {};
\draw[-] (distr)--(distr-|ell1)--(ell1);
\draw[-] (distr)--(distr-|ell2)--(ell2);
\draw[-,red] (t2.117) --++(-0.98,0) node[left]{};
\draw[-,blue] (t2.121) --++(-0.89,0) node[left]{};
\draw[-,orange] (t2.125) --++(-0.83,0) node[left]{};
\draw[-,green] (t2.129) --++(-0.90,0) node[left]{};
\draw[-,purple] (t2.133) --++(-1.00,0) node[left]{};
\draw[<-,color=red] (e1) --++(-0.7,0) node[above]{\small$M_1$};
\draw[<-,color=blue] (e11) --++(-0.7,0) node[above]{\small$M_2$};
\draw[<-,color=purple] (e2) --++(-0.7,0) node[above]{\small$M_K$};
\draw[->,color=red] (d1) --++(0.7,0) node[above]{\small$\hat{M}_1$};
\draw[->,color=blue] (d11) --++(0.7,0) node[above]{\small$\hat{M}_2$};
\draw[->,color=purple] (d2) --++(0.7,0) node[above]{\small$\hat{M}_K$};
\draw[-,color=red] ($(T.-180)+(0,+0.08)$)  -- ($(T.180-|t1.90)+(0,+0.08)$);
\draw[-,color=blue] ($(T.-180)+(0,+0.04)$)  -- ($(T.180-|t1.90)+(0,+0.04)$);
\draw[-,color=orange] ($(T.-180)+(0,+0.00)$)  -- ($(T.180-|t1.90)+(0,+0.00)$);
\draw[-,color=green] ($(T.-180)+(0,-0.04)$)  -- ($(T.180-|t1.90)+(0,-0.04)$);
\draw[-,color=purple] ($(T.-180)+(0,-0.08)$)  -- ($(T.180-|t1.90)+(0,-0.08)$);
\draw[->,color=red] (e1) --++(1.55,0) node[above]{};
\draw[->,color=blue] (e11) --++(1.55,0) node[above]{};
\draw[->,color=purple] (e2) --++(1.55,0) node[above]{};
\draw[->,color=red] (e3.east) --++(0.44,0) node[above]{};
\draw[->,color=blue] (e33.east) --++(0.44,0) node[above]{};
\draw[->,color=purple] (e4.east) --++(0.44,0) node[above]{};
\draw[<-,color=red] (d1.west) --++(-1.05,0) node[midway, above]{};
\draw[<-,color=blue] (d11.west) --++(-1.05,0) node[midway, above]{};
\draw[<-,color=purple] (d2.west) --++(-1.05,0) node[midway, above]{};
\draw[<-,color=red] (e7.west) --++(-0.39,0) node[above]{};
\draw[<-,color=blue] (e77.west) --++(-0.39,0) node[above]{};
\draw[<-,color=purple] (e8.west) --++(-0.39,0) node[above]{};
\draw[->,color=red] (e7.east) --++(0.32,0) node[above]{};
\draw[->,color=blue] (e77.east) --++(0.32,0) node[above]{};
\draw[->,color=purple] (e8.east) --++(0.32,0) node[above]{};
\draw[->,color=red] (e9.east) --++(0.4,0) node[above]{};
\draw[->,color=blue] (e99.east) --++(0.4,0) node[above]{};
\draw[->,color=purple] (e10.east) --++(0.4,0) node[above]{};
\draw[-,color=red] (e5.east) --++(0.24,0) node[above](tt1){};
\draw[-,color=blue] (e55.east) --++(0.24,0) node[above](tt2){};
\draw[-,color=purple] (e6.east) --++(0.24,0) node[above](tt3){};
\coordinate (A) at (8.3,1.7);
\coordinate (B) at (8.05,1.5);
\draw[-,color=red] (B) -- (A) node[above]{};
\coordinate (C) at (8.3,0.5);
\coordinate (D) at (8.05,0.3);
\draw[-,color=blue] (D) -- (C) node[above]{};
\coordinate (E) at (8.3,-1.3);
\coordinate (F) at (8.05,-1.5);
\draw[-,color=purple] (F) -- (E) node[above]{};
\coordinate (AA) at (7.98,1.7);
\coordinate (BB) at (8.2,1.25);
\draw [<-,color=red] (BB) to [out=90,in=10] (AA);
\coordinate (CC) at (7.98,0.5);
\coordinate (DD) at (8.2,0.05);
\draw [<-,color=blue] (DD) to [out=90,in=10] (CC);
\coordinate (EE) at (7.98,-1.3);
\coordinate (FF) at (8.2,-1.75);
\draw [<-,color=purple] (FF) to [out=90,in=10] (EE);
\end{tikzpicture}}
\caption{System model for WDM transmission under consideration is modeled as an interference channel with channel law $p(Y_1^n,Y_2^n,\ldots,Y_K^n|X_1^n,X_2^n,\ldots,X_K^n)$ and approximated model \eqref{eq:Kuserapprox}. The $k$-th user transmits message $M_k$ using an encoder $E_k$. After E-O conversion, propagation, O-E conversion, receiver DSP, a decoder $D_k$ is used. $A(t,z)$ represents the complex envelope of the optical field at time $t$ and distance $z$ from the transmitter, with $L$ being the length of the fiber.}
\label{fig:ICmodel}
\end{center}
\end{figure*}
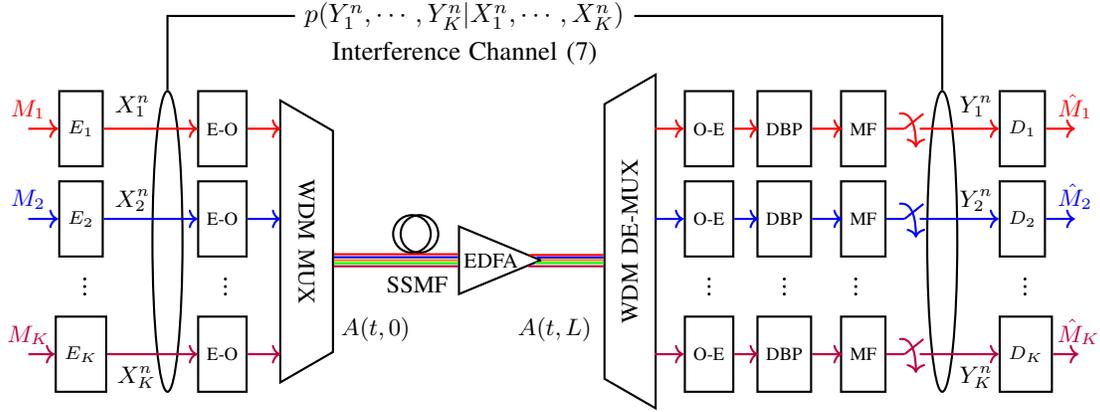

Information theory applied to optical communications studies transmission limits of such systems and has received increased interest in the last 10 years. Early works on the capacity limits of optical fibers were based on approximations involving low fiber nonlinearity~\cite{stark1999fundamental,mitra2001nonlinear,narimanov2002channel}. A capacity lower bound based on mismatched decoding~\cite{merhav1994information} was obtained in \cite{wegener2004effect} for a channel model incorporating XPM as well as FWM. More recently, the seminal work of \cite{essiambre2010capacity} stressed the importance of information theory in the studies of transmission limits over optical fiber channels. For WDM systems, the coupling between the different users resulting from XPM makes it a \emph{multiuser channel}, whose fundamental limits fall within the domain of \textit{multiuser information theory}. Unlike single-user information theory where channel capacity is the key quantity under study, the central object of interest in multiuser information theory is the \emph{capacity region}, i.e., the region of all simultaneously achievable rates of all the different users. {To the best of our knowledge, two works exist in the literature on multiuser information-theoretic characterizations of optical channels~\cite{taghavi2006multiuser,ghozlan2017models}.}

Despite the inherent multiuser nature of optical WDM channels, their information-theoretic analysis so far has been largely restricted to a single-user view focusing on the individual users. As such, optical WDM channels have rarely been truly analyzed from a multi-user perspective in the information-theoretic sense. For instance, \cite{agrell2015influence} examined the impact of different {behavioral assumptions} for the interfering users on the capacity of a specific user in the system. As a result of such assumptions, the characterisation of achievable information rates in \cite{agrell2015influence} is performed from a single-user perspective. Later, \cite{secondini2016scope} analyzed the capacity of a single user in the WDM system under the assumption that the interfering users transmit independent information at the same transmit power with the same modulation format. Under this behavioral model, it was shown that WDM capacity grows unbounded with power as opposed to Gaussian achievable information rates that exhibit a finite maximum.   

The aforementioned works \cite{agrell2015influence}, \cite{secondini2016scope} attempt to reduce the analysis of a multi-user problem to more familiar single-user problems by making various behavioral assumptions on the interfering users. 
However, such an approach is neither optimal from a single-user nor an overall WDM system perspective. In this paper, we deviate from the norm of a single-user information theoretic analysis of optical WDM channels, and investigate them from a multi-user information theoretic viewpoint. This better captures the rate contention amongst different WDM users and allows us to investigate the ultimate limits in a WDM system. In addition to achievable information rates for the different WDM users, capacity upper bounds are also of interest since they present impossibility results for the system under consideration. The analogue of capacity upper bounds in a multi-user framework is the notion of a capacity region outer bound (see Sec.~\ref{sec:prelims} for a precise definition), which is largely neglected in the literature on optical multi-user channels.\footnote{Capacity upper bounds for the single-user scenario do exist but are rare. The only known ones for a general NLSE (single-user waveform) channel are that of \cite{yousefi2015upper}, \cite{kramer2015upper} and \cite{keykhosravi2017tighter}.}

In the multi-user information theory literature, multiple one-to-one communications over a shared medium with crosstalk between the users is known as an \textit{interference channel}~\cite[Chapter 6]{el2011network}. Interference channels have attracted very little attention in the fiber optical communication literature. As noted earlier, only two papers exist on the topic, both of them for highly simplified channel models. The earliest of such work from 2006 was \cite{taghavi2006multiuser}, where the benefits of multi-user detection in WDM systems were analyzed by modeling it as a \emph{multiple access channel}, which is an interference channel with full receiver cooperation. More than ten years later, \cite{ghozlan2017models} studied a simplified interference channel model based on logarithmic perturbation ignoring group velocity dispersion across WDM bands and introduced the technique of \emph{interference focusing} to achieve the optimal high power pre-log factors.

With the aforementioned exceptions~\cite{taghavi2006multiuser,ghozlan2017models}, a study of the set of \emph{simultaneously achievable} rates that captures the contention amongst the different users accessing the optical channel transmission resources based on a realistic channel model is not available in the literature. Moreover, capacity region outer bounds are as of today also completely missing in the framework of optical multi-user channels.  

{ In this paper, we take a step in the direction of analyzing optical multi-user channels and study a simplified first-order perturbative multi-user model that considers both chromatic dispersion and Kerr nonlinearity.} We do not make assumptions such as full receiver cooperation as in \cite{taghavi2006multiuser} or negligible group velocity dispersion as in  \cite{ghozlan2017models}. {However, in order to keep the information-theoretic analysis tractable, we shall only consider the dominant contributions to XPM in a first-order regular perturbative model in our analysis.}
The main contributions of this paper are twofold: (i) We propose a novel outer bound on the {capacity region} of an optical multi-user channel where both the transmitters and the receivers are independently operated, and (ii) we obtain an achievable rate region by time-sharing between certain single-user strategies, and show that the latter can achieve better rate tuples compared to treating interference arising from other WDM users as noise (abbreviated \emph{TIN} henceforth). {Finally, we also perform a validation of the simplified model under consideration via split-step Fourier method (SSFM) simulations to discern the range of applicability (with respect to launch power) of the proposed techniques.} 

{ Parts of this work have been recently published in \cite{viswanathan2021wdm}, without theorem proofs or SSFM simulation results. Some flaws in the capacity curves of \cite{viswanathan2021wdm} were corrected in a recent conference publication~\cite{viswanathan2022wdm} by the authors, which has been incorporated into this extended version.}

\emph{Notation convention:} Random variables or random vectors are represented by upper-case letters, whereas their realizations are represented by the corresponding lower case characters. A length-$n$ block of random symbols is denoted by $X_k^n \triangleq (X_k[1],X_k[2],\ldots,X_{k}[n])$, where the subscript $k$ is a user index and the number within square brackets is a discrete time index. All logarithms in this paper are assumed to be with respect to base $2$, unless stated otherwise. Given a complex random variable $X$, we will denote its real part by $X^R$ and its imaginary part by $X^I$, i.e., $X=X^R+\jmath X^I$ with $\jmath = \sqrt{-1}$. Sets are denoted using calligraphic letters.

\emph{Paper Organization:} The channel model along with a review of some (network) information-theoretic preliminaries are first described in Sec.~\ref{sec:model}. 
Capacity region outer bounds are derived in Sec.~\ref{sec:Kuser}. Achievable rates for the individual users are computed in Sec.~\ref{sec:ach}. Sec.~\ref{sec:num} contains the numerical results and discussions. Finally, Sec.~\ref{sec:conc} concludes the paper.

\section{Preliminaries} \label{sec:model}

\subsection{System Model}
We study the $K$-user WDM system shown in Fig.~\ref{fig:ICmodel}, where the interference channel $p(Y_1^n,Y_2^n,\ldots,Y_K^n|X_1^n,X_2^n,\ldots,X_K^n)$ encompasses the electro-optical (E-O) conversion, WDM multiplexing, the physical channel, WDM demultiplexing, optical-electrical (O-E) conversion, single-channel DBP, matched filtering and symbol-rate sampling. Given that this is the first study on the capacity region for regular perturbative models from an interference channel viewpoint, we assume single-polarization transmission and ignore signal-noise interactions by studying a single span of standard single mode fiber (SSMF).

For such a single-mode fiber with Kerr nonlinearity and chromatic dispersion, the complex envelope of the optical field, $A(t,z)$, at time $t$ and distance $z$ from the transmitter is governed by the nonlinear Schr{\"o}dinger equation (NLSE)~\cite{agrawal2000nonlinear} 
\begin{align} \label{eq:nlse}
\pderiv{A(t,z)}{z} 
= \frac{\jmath}{2} \beta_2 \pderiv[2]{A(t,z)}{\tau} -\jmath \gamma |A(t,z)|^2 A(t,z)+W(t,z),
\end{align}
where $\tau=t-\beta_1 z$ is the shifted time reference of the moving pulse, with $\beta_1$ being the inverse of the group velocity. In \eqref{eq:nlse}, $\beta_2$ stands for the group velocity dispersion parameter, while $\gamma$ is the fiber nonlinearity parameter, with the second-last term on the right-hand side representing Kerr nonlinearity. The term $W(t,z)$ represents additive noise from the erbium doped fiber amplifier (EDFA) which ideally compensates for the fiber attenuation.

\begin{figure*}
\begin{center}
\scalebox{0.8}{\begin{tikzpicture}[thick]
\node (e1) at (-5.4,2.5) [rectangle, draw, minimum width=2.0cm, minimum height=0.6cm]{};
\node (e2) at (-3.4,2.5) [rectangle, draw, minimum width=2.0cm, minimum height=0.5cm]{\color{red}$|X_1[i-2]|^2$};
\node (e3) at (-1.4,2.5) [rectangle, draw, minimum width=2.0cm, minimum height=0.5cm]{\color{red}$|X_1[i-1]|^2$};
\node (e4) at (0.6,2.5) [rectangle, draw, minimum width=2.0cm, minimum height=0.5cm]{\color{red}$|X_1[i]|^2$};
\node (e5) at (2.6,2.5) [rectangle, draw, minimum width=2.0cm, minimum height=0.5cm]{\color{red}$|X_1[i+1]|^2$};
\node (e6) at (4.6,2.5) [rectangle, draw, minimum width=2.0cm, minimum height=0.5cm]{\color{red}$|X_1[i+2]|^2$};
\node (e7) at (6.6,2.5) [rectangle, draw, minimum width=2.0cm, minimum height=0.6cm]{};
\node (e8) at (8.6,2.5) [rectangle, draw, minimum width=2.0cm, minimum height=0.6cm]{};
\node (e9) at (-7.4,2.5) [rectangle, draw, minimum width=2.0cm, minimum height=0.6cm]{};
\node () at (-7.4,3.1) {$\cdots$};
\node () at (-5.4,3.1) {$i-3$};
\node () at (-3.4,3.1) {$i-2$};
\node () at (-1.4,3.1) {$i-1$};
\node () at (0.6,3.1) {$i$};
\node () at (2.6,3.1) {$i+1$};
\node () at (4.6,3.1) {$i+2$};
\node () at (6.6,3.1) {$i+3$};
\node () at (8.6,3.1) {$\cdots$};
\node () at (-9.2,2.5) {$k=1$};
\node () at (-9.2,0) {$k=3$};
\node () at (-9.2,-3) {$k=2$};
\node (c1) at (-3.4,1.25) [circle, draw, minimum height=0.35cm,thick,inner sep=0.2ex]{$f_2(\cdot)$};
\node (c2) at (-1.4,1.25) [circle, draw, minimum height=0.35cm,thick,inner sep=0.2ex]{$f_1(\cdot)$};
\node (c3) at (0.6,1.25) [circle, draw, minimum height=0.35cm,thick,inner sep=0.2ex]{$f_0(\cdot)$};
\node (c4) at (2.6,1.25) [circle, draw, minimum height=0.20cm,thick,inner sep=0.2ex]{\footnotesize $f_{-1}(\cdot)$};
\node (c5) at (4.6,1.25) [circle, draw, minimum height=0.20cm,thick,inner sep=0.2ex
]{\footnotesize$f_{-2}(\cdot)$};
\node (c6) at (2.6,-1) [circle, draw, minimum height=0.3cm,thick,inner sep=0.2ex]{$+$};
\node (c7) at (2.6,-2) [circle, draw, minimum height=0.3cm,thick,inner sep=0.2ex]{$\times$};
\node (c8) at (4.6,-2) [circle, draw, minimum height=0.3cm,thick,inner sep=0.2ex]{$+$};

\node (a1) at (-5.4,0) [rectangle, draw, minimum width=2.0cm, minimum height=0.6cm]{};
\node (a2) at (-3.4,0) [rectangle, draw, minimum width=2.0cm, minimum height=0.5cm]{\color{purple}$|X_3[i-2]|^2$};
\node (a3) at (-1.4,0) [rectangle, draw, minimum width=2.0cm, minimum height=0.5cm]{\color{purple}$|X_3[i-1]|^2$};
\node (a4) at (0.6,0) [rectangle, draw, minimum width=2.0cm, minimum height=0.5cm]{\color{purple}$|X_3[i]|^2$};
\node (a5) at (2.6,0) [rectangle, draw, minimum width=2.0cm, minimum height=0.5cm]{\color{purple}$|X_3[i+1]|^2$};
\node (a6) at (4.6,0) [rectangle, draw, minimum width=2.0cm, minimum height=0.5cm]{\color{purple}$|X_3[i+2]|^2$};
\node (a7) at (6.6,0) [rectangle, draw, minimum width=2.0cm, minimum height=0.6cm]{};
\node (a8) at (8.6,0) [rectangle, draw, minimum width=2.0cm, minimum height=0.6cm]{};
\node (a9) at (-7.4,0) [rectangle, draw, minimum width=2.0cm, minimum height=0.6cm]{};

\node (b1) at (-5.4,-3) [rectangle, draw, minimum width=2.0cm, minimum height=0.6cm]{};
\node (b2) at (-3.4,-3) [rectangle, draw, minimum width=2.0cm, minimum height=0.5cm]{\color{blue}$X_2[i-2]$};
\node (b3) at (-1.4,-3) [rectangle, draw, minimum width=2.0cm, minimum height=0.5cm]{\color{blue}$X_2[i-1]$};
\node (b4) at (0.6,-3) [rectangle, draw, minimum width=2.0cm, minimum height=0.5cm]{\color{blue}$X_2[i]$};
\node (b5) at (2.6,-3) [rectangle, draw, minimum width=2.0cm, minimum height=0.5cm]{\color{blue}$X_2[i+1]$};
\node (b6) at (4.6,-3) [rectangle, draw, minimum width=2.0cm, minimum height=0.5cm]{\color{blue}$X_2[i+2]$};
\node (b7) at (6.6,-3) [rectangle, draw, minimum width=2.0cm, minimum height=0.6cm]{};
\node (b8) at (8.6,-3) [rectangle, draw, minimum width=2.0cm, minimum height=0.6cm]{};
\node (b9) at (-7.4,-3) [rectangle, draw, minimum width=2.0cm, minimum height=0.6cm]{};

\draw[->,red] (e2.south) --++ (c1);
\draw[->,red] (e3.south) --++ (c2);
\draw[->,red] (e4.south) --++ (c3);
\draw[->,red] (e5.south) --++ (c4);
\draw[->,red] (e6.south) --++ (c5);
\draw[->,purple] (a2.north) --++ (c1);
\draw[->,purple] (a3.north) --++ (c2);
\draw[->,purple] (a4.north) --++ (c3);
\draw[->,purple] (a5.north) --++ (c4);
\draw[->,purple] (a6.north) --++ (c5);
\draw[<-] (c1) --++(-0.8,0) node[above]{$\jmath c_2^{2}$};
\draw[<-] (c2) --++(-0.8,0) node[above]{$\jmath c_2^{1}$};
\draw[<-] (c3) --++(-0.8,0) node[above]{$\jmath c_2^{0}$};
\draw[<-] (c4) --++(-0.9,0) node[above]{$\jmath c_2^{-1}$};
\draw[<-] (c5) --++(-0.9,0) node[above]{$\jmath c_2^{-2}$};
\draw[-,blue] (c1) --++(0.7,0) -| (a2.25);
\draw[-,blue] (c2) --++(0.7,0) -| (a3.25);
\draw[-,blue] (c3) --++(0.7,0) -| (a4.25);
\draw[-,blue] (c4) --++(0.8,0) -| (a5.20);
\draw[-,blue] (c5) --++(0.8,0) -| (a6.20);
\draw[->,blue] (a2.-25) --++ (0,-0.71) |- (c6.west);
\draw[->,blue] (a3.-25) --++ (0,-0.45) |- (c6.150);
\draw[->,blue] (a4.-25) --++ (0,-0.2) -| (c6.120);
\draw[->,blue] (a5.-20) --++ (0,-0.2) -| (c6.north);
\draw[->,blue] (a6.-20) --++ (0,-0.3) |- (c6.east);
\draw[->,blue] (c6) --++ (c7);
\draw[->,blue] (c7) --++ (c8);
\draw[->,blue] (b4.north) --++(0,0.7) |- (c7.west);
\draw[->,blue] (b4.north) --++(0,0.15) -| (c8);
\draw[<-] (c8) --++(0,0.4) node[above]{$N_2[i]$};
\draw[->,blue] (c8) --++(0.9,0) node[above]{$Y_2[i]$};
\end{tikzpicture}}
\caption{Illustration of the simplified channel model in \eqref{eq:Kuserapprox} for $K=3$ in Example 1 with $c_{2,1}^m=c_{2,3}^m=c_2^m$. The operation $f_m(\cdot)$ corresponds to $\jmath c_2^m \left(|X_1[i-m]|^2+|X_3[i-m]|^2\right)$ for $m \in \mathcal{M}$.}
\label{fig:ex1}
\end{center}
\end{figure*}
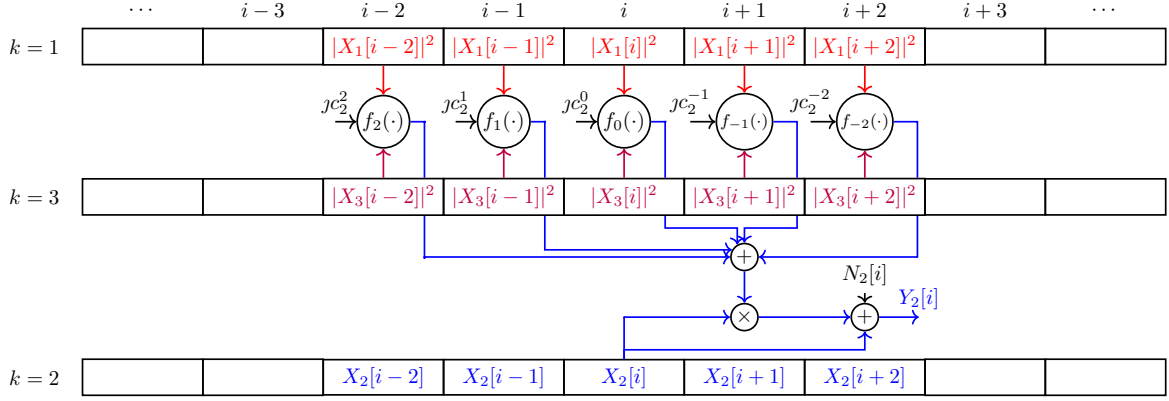

\subsection{Channel Model} \label{sec:chnmodel}

The output at the receiver of user$-k$, $k \in \{1,2,\ldots,K\}$, can be approximated using a first-order regular perturbative discrete-time model~\cite[eqs.~(59),(60)]{mecozzi2012nonlinear}, \cite[eqs.~(5),(7)]{dar2013properties}
\begin{align} \label{eq:Kuser}
{Y_{k}[i]}&{\approx} {X_{k}[i]+N_{k}[i]} \notag\\
&\phantom{w} {+\jmath \gamma \sum_{p=-\infty}^{\infty} X_{k}[i-p]} \notag\\
&\phantom{w}{\sum_{l=-\infty}^{\infty} \sum_{m=-\infty}^{\infty}  
{S_{k,w}^{p,l,m}} X_{w}[i-l] X_{w}^{*}[i-m]},
\end{align}
where $X_{k}[i]$ represents the input of user$-k$ at time instant $i \in \{1,2,\ldots,n\}$,
\begin{align}
\mathcal{W}_k\triangleq \{1,2,\ldots,K\}\setminus \{k\} \label{eq:Wnotation}
\end{align}
is the set of interferers for user $k$, $X_{w}[i]$ for $w \in \mathcal{W}_k$ are the inputs of the interfering users at instant $i$, $X_{w}[i-l]$ represents the corresponding input at a time lag\footnote{We use the convention that $X_w[i-l] = 0$ for $i-l < 1$ and $i-l > n$. In other words, $X_w[i-l] \neq 0$ only when the time index satisfies $ 1 \leq i-l \leq n$.} of $l$, and $\gamma$ is the fiber nonlinearity parameter from \eqref{eq:nlse}. The complex channel coefficients {$S_{k,w}^{p,l,m}$} are given in~\cite[eq.~(7)]{dar2013properties}, can be computed numerically, and depend on the properties of the optical link and the transmission parameters. {Specifically, they are given by
\begin{align}
S_{k,w}^{p,l,m} &= \int_{0}^L dz \: e^{-\alpha z} \int_{-\infty}^{\infty} dt \: g^{*}(z,t) \:g(z,t-pT) \notag\\
&\phantom{www} \times g^{*}(z,t-mT-\beta_2\Omega |k-w|z) \notag\\
&\phantom{www} \times g(z,t-lT-\beta_2 \Omega |k-w|z),
\end{align}
where $g(0,t)$ is the injected fundamental symbol waveform at time $t$ that becomes $g(z,t)$ when reaching point $z$ along the fiber, $L$ is the length of the link, $T$ is the symbol duration, $\beta_2$ is the group velocity dispersion parameter, $\Omega$ is the channel spacing and $\alpha$ is the loss coefficient.}
In \eqref{eq:Kuser}, $N_{k}[i]$ models amplified spontaneous emission (ASE) noise from the EDFAs. The ASE noise is  circularly symmetric complex Gaussian with mean zero and variance $\sigma_k^2$ per complex dimension.

We assume length-$n$ codewords $(x_k[1],x_k[2],\ldots,x_k[n])$ with maximum power constraints:
\begin{align}
&\max_{i \in \{1,2,\ldots,n\}} |x_{k}[i]|^2 \leq P_k, \: \forall \: k \in \{1,2,\ldots,K\}. \label{eq:peakpc}
\end{align}
In other words, $P_k$ represents a peak power constraint on the symbols transmitted by user-$k$, which is imposed on all the possible codewords in its codebook.
We note that the channel model specified by \eqref{eq:Kuser} is more realistic compared to the FWM-only model studied in \cite{agrell2015influence}, which assumes that both the dispersion and the nonlinearity are weak, and the generalized phase-matching condition is fulfilled~\cite{cappellini1991third}.

It is known from \cite[Figs. 4 and 5]{dar2016pulse} and \cite[eq.~(8)]{dar2013properties} that for few-span systems of relatively short lengths using lumped amplification, the largest contribution to the nonlinear interference (NLI) comes from the {$S_{k,w}^{0,m,m}$} terms in \eqref{eq:Kuser}, i.e., when only two time shifted sequences interact with each other. This corresponds to $p=0$ and $l=m$ in \eqref{eq:Kuser}, and is referred to as \emph{two-pulse collisions} in \cite{dar2016pulse}. In other words, the magnitudes {$|S_{k,w}^{0,m,m}|$} dominate over the terms corresponding to other values of the indices $p,l,m$. {This is illustrated in Fig.~\ref{fig:model_coeff_comp}, where the magnitudes of the coefficients are compared for a few different values of $p,l,m$ with $k=1, w=2$.} Furthermore, since physical channels do not have infinite memory, we truncate the sums on $p,l,m$ in \eqref{eq:Kuser} to the set
\begin{align}
\mathcal{M}\triangleq\{-M,-M+1,\ldots,M-1,M\}. \label{eq:Mnotation}
\end{align}
This results in the following approximate model:
\begin{align} \label{eq:Kuserapprox}
{Y_{k}[i]} &{\approx X_{k}[i]+N_{k}[i]} \notag\\
&\phantom{w} {+\jmath \gamma X_{k}[i] \sum_{m\in\mathcal{M}}  \sum_{\substack{w\in\mathcal{W}_k}}   {S_{k,w}^{0,m,m}} |X_{w}[i-m]|^2} \notag\\
&{= X_{k}[i]\biggl(1+\jmath  \sum_{m\in\mathcal{M}} \sum_{\substack{w\in\mathcal{W}_k}}  {c_{k,w}^m} |X_{w}[i-m]|^2\biggr)+N_{k}[i]},
\end{align}
where we have defined 
\begin{align}
{c_{k,w}^m \triangleq \gamma S_{k,w}^{0,m,m}}  \label{eq:coeff_abbrev}
\end{align}
for compactness. {The coefficients $c_{k,w}^m$} (computed along the direction $p=0$ and $l=m$ in \eqref{eq:Kuser} using \cite[eq.~(8)]{dar2013properties}) are known to be nonnegative reals, i.e.,
\begin{align}
{c_{k,w}^m \geq 0}.
\label{c.positive}
\end{align}
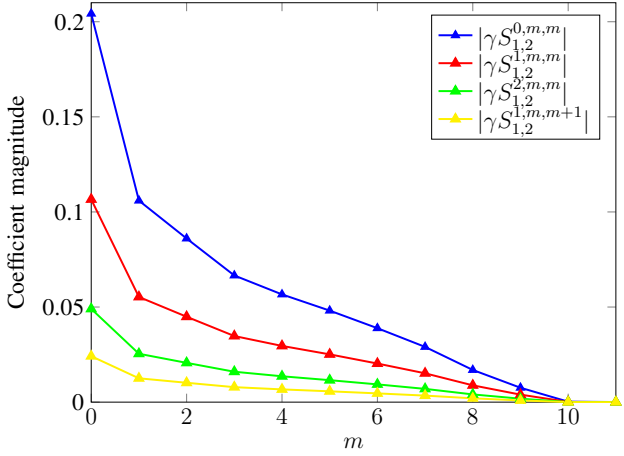
\begin{figure}[!t]
\vspace{-3cm}
\begin{center}
\scalebox{0.75}{
%
%
\definecolor{mycolor1}{rgb}{0.00000,0.44700,0.74100}%
\definecolor{mycolor2}{rgb}{0.85000,0.32500,0.09800}%
\begin{tikzpicture}

\begin{axis}[%
every axis/.append style={font=\large},
width=1.05\columnwidth,
height=0.8\columnwidth,
scale only axis,
xmin=0,
xmax=11,
xlabel={$m$},
ymin=0,
ymax=0.21,
ytick={0,0.05,0.1,0.15,0.2},
yticklabels={0,0.05,0.1,0.15,0.2},
ylabel={$\text{Coefficient magnitude}$},
ylabel near ticks,
xlabel near ticks,
axis background/.style={fill=white},
legend style={
legend pos=north east,legend cell align=left,row sep=-0.5ex,grid style={dashed},font=\large
}
]

\addplot [color=blue, mark = triangle*, mark size=2, line width=1.0pt]
  table[row sep=crcr]{%
0	0.20424\\
1	0.10596\\
2	0.0860400000000006\\
3	0.0665999999999993\\
4	0.0566399999999998\\
5	0.0481200000000008\\
6	0.0388800000000007\\
7	0.0290400000000002\\
8	0.0169200000000007\\
9	0.00755999999999979\\
10	0.000239999999999796\\
11	5.99999999995049e-05\\
};
\addlegendentry{$|\gamma S_{1,2}^{0,m,m}|$}

\addplot [color=red, mark = triangle*, mark size=2.5, line width=1.0pt]
  table[row sep=crcr]{%
0	0.10659199556541\\
1	0.0553000776053215\\
2	0.0449039135254989\\
3	0.0347582594235033\\
4	0.029560177383592\\
5	0.0251136252771619\\
6	0.0202913082039911\\
7	0.0151558536585366\\
8	0.00883047671840355\\
9	0.00394553215077605\\
10	0.000125254988913525\\
11	3.13137472283814e-05\\
};
\addlegendentry{$|\gamma S_{1,2}^{1,m,m}|$}

\addplot [color=green, mark = triangle*, mark size=2.5, line width=1.0pt]
  table[row sep=crcr]{%
0	0.0490221286031042\\
1	0.0254327494456763\\
2	0.0206515077605321\\
3	0.0159854767184035\\
4	0.0135948558758315\\
5	0.011549866962306\\
6	0.0093320620842572\\
7	0.00697024390243902\\
8	0.00406117516629712\\
9	0.00181456762749446\\
10	5.76053215077605e-05\\
11	1.44013303769401e-05\\
};
\addlegendentry{$|\gamma S_{1,2}^{2,m,m}|$}

\addplot [color=yellow, mark = triangle*, mark size=2.5, line width=1.0pt]
  table[row sep=crcr]{%
0	0.0241578541027418\\
1	0.0125331287736316\\
2	0.0101769573394041\\
3	0.00787756112045927\\
4	0.00669947540334554\\
5	0.00569171533207958\\
6	0.00459879243788974\\
7	0.00343490052459665\\
8	0.00200132633871127\\
9	0.000894209640700782\\
10	2.83876076412947e-05\\
11	7.09690191032367e-06\\
};
\addlegendentry{$|\gamma S_{1,2}^{1,m,m+1}|$}

\end{axis}

\begin{axis}[%
width=5.833in,
height=4.375in,
at={(0in,0in)},
scale only axis,
xmin=0,
xmax=1,
ymin=0,
ymax=1,
axis line style={draw=none},
ticks=none,
axis x line*=bottom,
axis y line*=left
]
\end{axis}
\end{tikzpicture}%
}
\end{center}
\caption{ Absolute value of channel coefficients $|\gamma S_{1,2}^{0,m,m}|, |\gamma S_{1,2}^{1,m,m}|, |\gamma S_{1,2}^{2,m,m}|, |\gamma S_{1,2}^{1,m,m+1}|$ involved in the model \eqref{eq:Kuser}. Since the coefficients are symmetric about $m=0$, only the positive time indices are shown. It is seen that the coefficients corresponding to $|S_{1,2}^{0,m,m}|$ dominate over the rest corresponding to other values of $p,l,m$.
}
\label{fig:model_coeff_comp}
\end{figure}

Notice that in \eqref{eq:Kuserapprox}, only $M$ symbols before and after the current time instance contribute to the nonlinear interference, as opposed to the infinite summations involved in \eqref{eq:Kuser}. This results in a finite-memory channel similar in structure to the heuristic model introduced and studied from a single-user point of view in \cite{agrell2014capacity}. 
We shall work with the model in \eqref{eq:Kuserapprox} in the sequel. { The same model was also employed in \cite[eq.~(48)]{taghavi2006multiuser} for XPM, albeit in a multiple access channel context as opposed to the interference channel under consideration here. The given model was also analyzed with a view towards estimating the variance of the nonlinear interference terms in \cite{dar2013properties}.}

{We note that the regular perturbative model in \eqref{eq:Kuser} (and consequently the simplified model in \eqref{eq:Kuserapprox}) is not energy preserving, as has been observed in \cite[Sec.~VI]{vannucci2002rp}. In other words, in the absence of additive noise, the channel appears to behave as an amplifier with an absolute value of gain larger than one (see also the discussion on the black dotted curve in Fig.~\ref{fig:OBsnr} in Sect.~\ref{sec:num}). This is in contrast to the NLSE channel, which is a conservative system. Nevertheless, the model in (6) has been adopted for XPM in the literature (\cite[eq.~(48)]{taghavi2006multiuser} and \cite{dar2013properties}), as mentioned earlier.}

\begin{example}[3 WDM channels] \label{ex:1}
Consider the case of $K=3$ users, the user of interest being $k=2$, and a single-sided channel memory of $M=2$ symbols. Assume for simplicity that $c_{2,1}^m=c_{2,3}^m=c_2^m$. In this case, the received symbols for user-2 are given by
\begin{align}
Y_{2}[i]&=X_{2}[i]+N_{2}[i] \notag\\
&+\jmath X_{2}[i] \sum_{m=-2}^{2} c_2^m \bigl(|X_{1}[i-m]|^2+|X_{3}[i-m]|^2\bigr).
\end{align}
This is pictorially represented in Fig.~\ref{fig:ex1}.
\end{example}

\subsection{Information-theoretic Preliminaries} \label{sec:prelims}
In this section, we review some relevant information-theoretic notions for the $K-$user model in Fig.~\ref{fig:ICmodel}, modeled by \eqref{eq:Kuserapprox}. An $(n,2^{nR_1},2^{nR_2},\ldots,2^{nR_K})$ code for this channel consists of $K$ message sets $\{1,2,\ldots,2^{nR_k}\}$ for $k \in \{1,2,\ldots,K\}$, $K$ encoders where $E_k$ maps a message $M_k \in \{1,2,\ldots,2^{nR_k}\}$ into a codeword $X_k^n(M_k)$, along with the decoders. The messages $M_k$ are assumed to be equally likely on their respective alphabets $\{1,2,\ldots,2^{nR_k}\}$ for all $k \in \{1,2,\ldots,K\}$, where $R_k$ is the transmission rate of user $k$.

At the receiver, $K$ decoders $D_k$ assign an estimate $\hat{M}_k$ (or an error message) to each received sequence $Y_k^n$. The probability of error is defined as
\begin{align} \label{eq:errorprob}
P_e \triangleq \textup{Pr}\{(\hat{M}_1(Y_1^n),\ldots,\hat{M}_K(Y_K^n)) \neq (M_1,\ldots,M_K)\}.
\end{align}
{\begin{remark}
While the error probability definition in \eqref{eq:errorprob} depends on the decisions of all the decoders, we emphasize here that the $K$ decoders do not cooperate, as shown at the receiver side of Fig.~\ref{fig:ICmodel}.
\end{remark}}

Using the above definitions, we now formally define certain important quantities. These quantities will be explained later using an example.
\begin{definition}[Achievability]\label{achievability}
A rate tuple $(R_1,R_2,\ldots,R_K)$ is said to be \emph{achievable} if there exists a sequence of $(n,2^{nR_1},2^{nR_2},\ldots,2^{nR_K})$ codes such that $\lim_{n \to \infty} P_e = 0$.
\end{definition}

\begin{definition}[Capacity Region]\label{CK}
The capacity region $\mathcal{C}_K$ is defined as the closure of the set of all achievable rate tuples $(R_1,R_2,\ldots,R_K)$.
\end{definition}
The capacity region in Definition~\ref{CK} is a collection of all rate tuples that are achievable as per Definition~\ref{achievability}. 
{\begin{remark}
Note that Definition~\ref{CK} is an \emph{operational definition} of the capacity region as commonly used in multi-user information theory~\cite{el2011network}. An optimization over the joint probability distributions of the inputs is implicit in such definitions. This is in contrast to the single-user channel capacity definition often found in the literature (i.e., $\max_{p(x)} I(X;Y)$), where the optimization over the input distribution is made explicit.
\end{remark}}

\begin{definition}[Single-user Capacity]\label{Ck}
The single-user capacity is defined as: 
\begin{align}
C_k \triangleq \max_{(R_1,R_2,\ldots,R_K) \in \mathcal{C}_K} R_k.
\end{align}
\end{definition}
The single-user capacity for user $k$ in Definition~\ref{Ck} can then be interpreted as the largest achievable rate $R_k$, obtained while the rates of all other users are also achievable according to Definition~\ref{achievability}.

\begin{definition}[Capacity Inner/Outer Bounds]\label{Cinout}
A region $\mathcal{C}_{\textup{in}}$ is said to be an inner bound to $\mathcal{C}_K$ if every rate tuple $(R_1,R_2,\ldots,R_K) \in \mathcal{C}_{\textup{in}}$ is achievable. A region $\mathcal{C}_{\textup{out}}$ is said to be an outer bound to $\mathcal{C}_K$ if every achievable rate tuple satisfies $(R_1,R_2,\ldots,R_K) \in \mathcal{C}_{\textup{out}}$.
\end{definition}

The inner bound in Definition~\ref{Cinout} is also often called an \textit{achievable region}. This inner bound is a subset of the capacity region whose interior is entirely achievable. The definition of the outer bound in Definition~\ref{Cinout} is such that $\mathcal{C}_{\textup{out}}$ contains all the achievable rate tuples, i.e., it contains the capacity region. However, unless it is a perfectly tight bound, $\mathcal{C}_{\textup{out}}$ will also contain rate tuples that are not achievable. 

{\begin{remark}The inner and outer bounds in Definition~\ref{Cinout} are generalizations of the familiar notions of single-user capacity and lower/upper bounds. It follows from the above definitions that $\mathcal{C}_{\textup{in}} \subseteq \mathcal{C}_{K} \subseteq \mathcal{C}_{\textup{out}}$. For the special case of $K=1$, the operation of containment $\subseteq$ is replaced by an inequality $\leq$ and the sets $\mathcal{C}$ become scalars.
\end{remark}}

We next review the notion of time-sharing which is commonly used to obtain inner bounds in multi-user information theory.
\begin{definition}[Time Sharing]\label{timeshare}
Given any two achievable rate tuples $(R_{1}',R_{2}',\ldots,R_{K}')$ and $(R_{1}'',R_{2}'',\ldots,R_{K}'')$, time sharing between them results in the rate tuple 
\begin{align}
&(R_{1\lambda},R_{2\lambda},\ldots,R_{K\lambda}) \notag\\
&= (\lambda R_{1}'+\bar{\lambda}R_{1}'',\lambda R_{2}'+\bar{\lambda}R_{2}'',\ldots,\lambda R_{K}'+\bar{\lambda}R_{K}''), \label{eq:timeshare}
\end{align}
where $\lambda \in [0,1]$ and $\bar{\lambda}=(1-\lambda)$.
\end{definition}
The rate tuple given by \eqref{eq:timeshare} is achievable as well. A proof of this statement is given for instance in \cite[Proposition 4.1]{el2011network}.

\begin{example}[Information-theoretic quantities]\label{ex:it}
The information theoretic concepts just described are illustrated in Fig.~\ref{fig:capillus}. The shaded region in red represents the capacity region $\mathcal{C}_K$. $C_1$ and $C_2$ represent the single-user capacities of the two users. Notice that when user$-1$ achieves its single-user capacity $C_1$, it is possible to obtain a nonzero rate for user$-2$. In other words, the rate of user$-2$ can be increased up to the corner point of the pentagon (marked as $P_1$) without reducing the rate of user$-1$. The shaded region in purple marked $\mathcal{C}_{\textup{in}}$ as well as the shaded region in blue marked $\mathcal{C}_{\textup{in}}'$ are inner bounds, while the region marked $\mathcal{C}_{\textup{out}}$ is an outer bound to the capacity region $\mathcal{C}_K$. In this example, $\mathcal{C}_{\textup{in}} \subseteq \mathcal{C}_{\textup{in}}' \subseteq \mathcal{C}_{K} \subseteq \mathcal{C}_{\textup{out}}$, and the outer bound is not tight, resulting in nonachievable rates (like $B$) being included in $\mathcal{C}_{\textup{out}}$. The region $\mathcal{C}_{\textup{out}}$ defines an inadmissible region, in that it is impossible to achieve any rate pairs outside $\mathcal{C}_{\textup{out}}$. The dotted line illustrates time sharing, where every point on the line segment joining two achievable rate pairs is achievable as well -- this line is traced by varying $\lambda$ from \eqref{eq:timeshare} in the interval $[0,1]$.
\begin{figure}
\begin{center}
\scalebox{0.925}{\definecolor{mycolor1}{rgb}{0.76078,0.76078,0.76078}%
\definecolor{mycolor2}{rgb}{0.49412,0.49412,0.96863}%
\definecolor{mycolor3}{rgb}{0.92900,0.69400,0.12500}%
\definecolor{mycolor4}{rgb}{0.80000,0.42353,0.80000}%
\begin{tikzpicture}[thick]
    \draw[->] (0,0) -- (6,0) coordinate (xaxis);
    \draw[->] (0,0) -- (0,6) coordinate (yaxis);
    \draw[fill=red!55!white]
    (0,0,0) -- (4,0,0) -- (4,2,0) -- (2,4,0) -- (0,4,0) -- cycle;
    \draw[thick] (0,0,0) -- (5,0,0) -- (5.5,2,0) -- (2,5.5,0) -- (0,5,0) -- cycle;
    \draw[fill=mycolor2] 
    (0,0,0) -- (3,0,0) -- (3,1.5,0) -- (1.5,3,0) -- (0,3,0) -- cycle;
    \draw[fill=mycolor4]
    (0,0,0) -- (1.5,0,0) -- (1.5,1.5,0) -- (0,1.5,0) -- cycle;
    \draw[thick,dashed] (0,1.5) -- (1.5,0);
    \node[below,yshift=-1ex] at (xaxis) {$R_1 \:  \textup{(bits/sym)}$};
    \node[left,xshift=15ex] at (yaxis) {$R_2 \: \textup{(bits/sym)}$};
    \node[below left] {$O$};
    \node at (0.9,1.1) {$\mathcal{C}_{\textup{in}}$};
    \node at (1.75,1.75) {$\mathcal{C}_{\textup{in}}'$};
    \node at (2.7,2.7) {$\mathcal{C}_{K}$};
    \node at (3.5,3.5) {$\mathcal{C}_{\textup{out}}$};
    \fill (4,0) circle(0.07cm);
    \fill (0,4) circle(0.07cm);
    
    \fill (2,4) circle(0.07cm);
    \node at (4,-0.3) {$C_1$};
    \node at (-0.35,4) {$C_2$};
    \node at (4.35,2) {$P_1$};\fill (4,2) circle(0.07cm);
    \node at (2.35,4) {$P_2$};
    \node at (4.4,0.75) {$B$};\node at (4.4,1.1) {\large $\times$};
    \node at (1.3,0.6) {$A_1$};\fill (1.3,0.85) circle(0.07cm);
    \node at (1.0,2.5) {$A_2$};\fill (1.0,2.25) circle(0.07cm);
    \node at (1.2,3.25) {$A_3$};\fill (1.2,3.5) circle(0.07cm);
    \node[draw=black,rounded corners,fill=white,inner sep=1ex] at (5,5) {Inadmissible Region};
\end{tikzpicture}}
\caption{Illustration of the notions of capacity region inner and outer bounds. The shaded region in red represents the capacity region $\mathcal{C}_K$. $C_1$ and $C_2$ represent the single-user capacities of the two users. The shaded region in purple marked $\mathcal{C}_{\textup{in}}$ is an inner bound to $\mathcal{C}_K$, and so is the region in blue marked $\mathcal{C}_{\textup{in}}'$. The region marked $\mathcal{C}_{\textup{out}}$ is an outer bound to $\mathcal{C}_K$. Points such as $A_1$, $A_2$, and $A_3$ are achievable, while points like $B$ which fall outside $\mathcal{C}_K$ are not achievable. The region $\mathcal{C}_{\textup{out}}$ defines an inadmissible region, in that it is impossible to achieve any rate pairs outside $\mathcal{C}_{\textup{out}}$. Note that every point on the dotted line segment joining two achievable rate pairs is achievable as well via time sharing.}
\label{fig:capillus}
\end{center}
\end{figure}
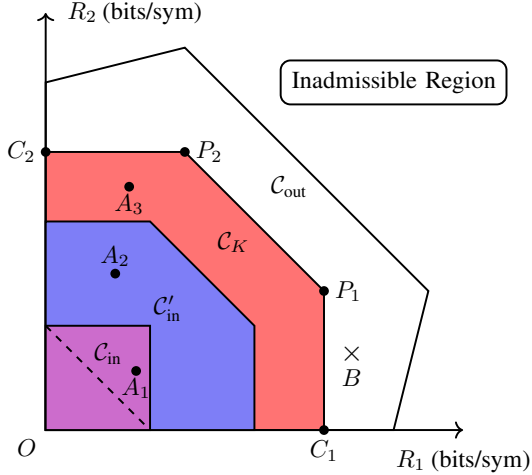
\end{example} 

We next discuss a simple strategy to obtain an inner bound for the channel model under consideration in \eqref{eq:Kuserapprox}. We resort to the most commonly used approach in WDM systems, that involves each user treating nonlinear interference as Gaussian noise (TIN). {In the context of expression \eqref{eq:Kuserapprox}, this approach means that the entire NLI term}
\begin{align}\label{NLI.TIN}
{\jmath X_{k}[i] \sum_{m\in\mathcal{M}} \sum_{\substack{w\in\mathcal{W}_k}}   {c_{k,w}^m}|X_{w}[i-m]|^2}
\end{align}
{is treated as Gaussian noise}. 
The achievable rate for a specific user via TIN is given by $L_k^{\textup{TIN}}$ (which is a lower bound to the single-user capacity\footnote{The lower bound is derived by treating the NLI as Gaussian noise (like the ASE noise term $N_k[i]$) and using the entropy power inequality, similar to Appendix~\ref{sec:appach}.} $C_k$), specified by:
\begin{align} \label{eq:lbTIN}
C_k \geq L_k^{\textup{TIN}} \triangleq \log\left(1+\frac{P_k}{2(\sigma_k^2+\sigma_{\textup{NLI}}^2)e}\right),
\end{align}
where $\sigma_{\textup{NLI}}^2$ is the variance of the term \eqref{NLI.TIN}. All users can simultaneously achieve \eqref{eq:lbTIN}, resulting in a square-shaped inner bound as shown in the purple region $\mathcal{C}_{\textup{in}}$ in Fig.~\ref{fig:capillus}.
The TIN strategy yields (single-user) achievable rates that exhibit a peaky behaviour as a function of power, often referred to as the nonlinear Shannon limit~\cite{splett1993ultimate,mitra2001nonlinear,ellis2009approaching}.

\section{Main Results}
The main results in this paper are organized into three subsections. One of our key contributions, a novel outer bound on the capacity region, is discussed in Sec.~\ref{sec:Kuser}. Next, Sec.~\ref{sec:ach} describes an inner bound on the capacity region obtained via time-sharing between certain single-user strategies. Finally, Sec.~\ref{sec:num} contains the numerical results and discussions on these capacity bounds.

\subsection{Capacity Region Outer Bounds} \label{sec:Kuser}
Here we obtain an outer bound on the capacity region $\mathcal{C}_K$ (Theorem~\ref{thm:OBgenieK}, ahead) using genie-aided techniques~\cite{etkin2008gaussian}. The following lemma will prove useful towards this end.

\begin{lemma} \label{lem:optinterf}
For all interferers $w\in\mathcal{W}_k$ with peak power constraints $P_w$ in \eqref{eq:peakpc} and $1 \leq i-m \leq n$, we have
\begin{align}
{\sum_{m\in\mathcal{M}} {c_{k,w}^m} |x_{w}[i-m]|^2 \leq P_w \left(\sum_{m\in\mathcal{M}} {c_{k,w}^m}\right)}\!, \label{eq:lem1ineq}
\end{align}
where $\mathcal{W}_k$, $\mathcal{M}$, and  ${c_{k,w}^m}$ are given by \eqref{eq:Wnotation}, \eqref{eq:Mnotation}, and \eqref{eq:coeff_abbrev}, resp. Equality is obtained in \eqref{eq:lem1ineq} with a choice of
\begin{align}\label{eq:constampl}
|x_{w}[i-m]|^2 = P_w.
\end{align}
\end{lemma}
\begin{proof}
See Appendix~\ref{app:lem:genieK}.
\end{proof}
Note that \eqref{eq:constampl} involves constant-amplitude signaling for the interferers $w\in\mathcal{W}_k$.
{
\begin{remark}
The conclusion from Lemma~\ref{lem:optinterf} is that all the interferers should always transmit at the maximum possible power. However, in a multi-span situation, there would be signal-noise interactions due to the inline amplifiers, and we do not expect this result to generalize.
\end{remark}
}
Our outer bound is stated next.
\begin{theorem} \label{thm:OBgenieK}
An outer bound $\mathcal{C}_{\textup{out}}$ on the capacity region $\mathcal{C}_K$ of the interference channel in \eqref{eq:Kuserapprox} is specified by the set of $(R_1,R_2,\ldots,R_K)$ tuples such that
\begin{align} \label{eq:ob1}
&R_k \leq U_k, \: \forall \:  k \in \{1,2,\ldots,K\},
\end{align}
where
\begin{align} \label{eq:ob2}
{U_k} &{\triangleq \log\left(1+\frac{P_{k}}{2\sigma_k^2}\left(1+\left(\sum_{\substack{w\in\mathcal{W}_k}}  P_w \sum_{m\in\mathcal{M}} {c_{k,w}^m}\right)^2\right) \right)},
\end{align}
and $P_k$ are peak power constraints in \eqref{eq:peakpc}, $\mathcal{W}_k$ and $\mathcal{M}$ are defined in \eqref{eq:Wnotation} and \eqref{eq:Mnotation} respectively, while ${c_{k,w}^m}$ is given in \eqref{eq:coeff_abbrev}.
\end{theorem}
\begin{proof}
See Appendix~\ref{app:genieK}.
\end{proof}

The proof of Theorem~\ref{thm:OBgenieK} involves the identification of constant-amplitude signaling for the interferers $w\in\mathcal{W}_k$ (see \eqref{eq:constampl}) to be the best strategy with regards to maximizing the rate of user-$k$. Lemma~\ref{lem:optinterf} forms the basis for our achievability scheme discussed next.

\subsection{Capacity Region Inner Bounds} \label{sec:ach}
We first develop a (single-user) capacity lower bound for each individual user $k \in \{1,2,\ldots,K\}$. We then time-share between such single-user achievability strategies to obtain a capacity region inner bound $\mathcal{C}_{\textup{in}}$.
Towards this end, consider the channel output for user-$k$, $k \in \{1,2,\ldots,K\}$ in \eqref{eq:Kuserapprox}.
Suppose the interferer symbols are chosen\footnote{This approach of obtaining lower bounds for the channel of interest by choosing the interferer behaviour was also used in \cite{agrell2015influence}.} as in \eqref{eq:constampl} from Lemma~\ref{lem:optinterf}.
This results in the following \textit{memoryless} single-user channel:
\begin{align}
{Y_{k}[i]} &{= X_{k}[i]\left(1+\jmath \sum_{m\in\mathcal{M}} {c_{k,w}^m} \sum_{\substack{w\in\mathcal{W}_k}}  P_w \right)+N_{k}[i]}. \label{eq:shamai}
\end{align}
Notice that \eqref{eq:shamai} is a complex AWGN channel with a peak power constraint on the input, which has been extensively studied in the information theory literature~\cite{smith1971information,shamai1995capacity,sharma2010transition,thangaraj2017capacity,dytso2019capacity}. It is known that the capacity achieving input distribution for this channel is discrete in amplitude with uniform phase. No closed form expressions exist for the capacity of the channel, but the number of mass points for the amplitude of the capacity achieving input distribution as a function of the signal-to-noise ratio have been characterized~\cite{shamai1995capacity}.

For our purposes of computing an achievable rate for user-$k$ in \eqref{eq:shamai} (under constant-amplitude signaling for the interferers), we resort to the lower bounding technique used in~\cite[eq.~(38)]{shamai1995capacity}, based on the entropy power inequality. 
We have the following theorem that gives a lower bound on the single-user capacity for user-$k$, $k \in \{1,2,\ldots,K\}$.
\begin{theorem} \label{thm:ach}
The single-user capacity of user-$k$ under peak input power constraints is lower bounded as:
\begin{align} \label{eq:lbepi0}
C_k \geq L_k, \: \forall \:  k \in \{1,2,\ldots,K\},
\end{align}
where
\begin{align} \label{eq:lbepi}
{L_k} &{\triangleq \log\left(1+\frac{P_{k}}{2\sigma_k^2 e}\left(1+\left(\sum_{\substack{w\in\mathcal{W}_k}}  P_w \sum_{m\in\mathcal{M}} {c_{k,w}^m}\right)^2\right) \right)}.
\end{align}
\end{theorem}
\begin{proof}
See Appendix~\ref{sec:appach}.
\end{proof}

Theorem~\ref{thm:ach} defines an achievable rate for user-$k$ in the model specified by \eqref{eq:Kuserapprox}. The rate in \eqref{eq:lbepi} is achieved when all the interferers do constant-amplitude signaling, i.e., the interference symbols satisfy $|x_{w}[j]|^2 = P_w$, $\forall \:w\in\mathcal{W}_k, \: 1 \leq j \leq n$, while user-$k$ uses symbols distributed according to \cite[eq.~(30)]{shamai1995capacity} with the phase of $X_k$ being uniform on $[-\pi,\pi]$ and independent of the amplitude $|X_k|=R_k$ that has probability density function:
\begin{align}\label{pdf.R}
p_{R_k}(r)=
\begin{cases}
\frac{2r}{P_k}, 0 \leq r \leq \sqrt{P_k}\\
0, \textup{elsewhere}.
\end{cases}
\end{align}

The strategy of constant-amplitude signaling for the interferers $w\in\mathcal{W}_k$, along with the scheme based on \eqref{pdf.R} for user-$k$, together define $K$ achievable rate tuples on the $K$-dimensional plane. Time-sharing between such achievable rate tuples (see \eqref{eq:timeshare}) yields an inner bound $\mathcal{C}_{\textup{in}}$ for the channel in \eqref{eq:Kuserapprox}. A comparison between the TIN inner bound in \eqref{eq:lbTIN}, the outer bound $\mathcal{C}_{\textup{out}}$ in Theorem~\ref{thm:OBgenieK} and the inner bound $\mathcal{C}_{\textup{in}}$ obtained by time-sharing between rate tuples resulting from Theorem~\ref{thm:ach} follows next.

\begin{remark}
The upper and lower bounds in Theorems~\ref{thm:OBgenieK} and \ref{thm:ach} only differ by the factor of $e$ that appears in the denominator of the argument of the logarithm in \eqref{eq:lbepi}. Asymptotically, in the limit of high signal-to-noise ratio, the difference between the right-hand-sides of \eqref{eq:ob2} and \eqref{eq:lbepi} is simply $\log_2(e)$ bits. This is similar in spirit to the constant-gap capacity results (such as the ``half-bit theorem'') for linear Gaussian interference channels~\cite{etkin2008gaussian}.
\end{remark}

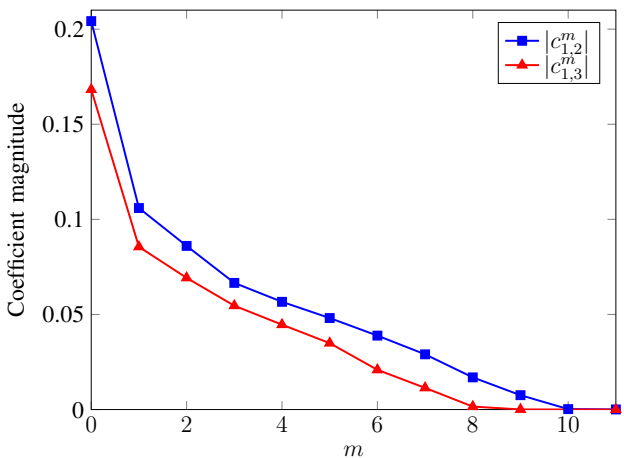
\begin{figure}[!b]
\vspace{-3cm}
\begin{center}
\scalebox{0.75}{
%
%
\definecolor{mycolor1}{rgb}{0.00000,0.44700,0.74100}%
\definecolor{mycolor2}{rgb}{0.85000,0.32500,0.09800}%
\begin{tikzpicture}

\begin{axis}[%
every axis/.append style={font=\large},
width=1.05\columnwidth,
height=0.8\columnwidth,
scale only axis,
xmin=0,
xmax=11,
xlabel={$m$},
ymin=0,
ymax=0.21,
ytick={0,0.05,0.1,0.15,0.2},
yticklabels={0,0.05,0.1,0.15,0.2},
ylabel={$\text{Coefficient magnitude}$},
ylabel near ticks,
xlabel near ticks,
axis background/.style={fill=white},
legend style={
legend pos=north east,legend cell align=left,row sep=-0.5ex,grid style={dashed},font=\large
}
]

\addplot [color=blue, mark = square*, mark size=2, line width=1.0pt]
  table[row sep=crcr]{%
0	0.20424\\
1	0.10596\\
2	0.0860400000000006\\
3	0.0665999999999993\\
4	0.0566399999999998\\
5	0.0481200000000008\\
6	0.0388800000000007\\
7	0.0290400000000002\\
8	0.0169200000000007\\
9	0.00755999999999979\\
10	0.000239999999999796\\
11	5.99999999995049e-05\\
};
\addlegendentry{$|c_{1,2}^m|$}

\addplot [color=red, mark = triangle*, mark size=2.5, line width=1.0pt]
  table[row sep=crcr]{%
0	0.168240000000001\\
1	0.0855599999999992\\
2	0.0692400000000006\\
3	0.0546000000000006\\
4	0.0446399999999993\\
5	0.0349199999999996\\
6	0.02088\\
7	0.0114000000000001\\
8	0.00155999999999956\\
9	5.99999999995049e-05\\
11	5.99999999906231e-07\\
};
\addlegendentry{$|c_{1,3}^m|$}

\end{axis}

\begin{axis}[%
width=5.833in,
height=4.375in,
at={(0in,0in)},
scale only axis,
xmin=0,
xmax=1,
ymin=0,
ymax=1,
axis line style={draw=none},
ticks=none,
axis x line*=bottom,
axis y line*=left
]
\end{axis}
\end{tikzpicture}%
}
\end{center}
\caption{ Absolute value of channel coefficients $|c_{1,2}^m|, |c_{1,3}^m|$ involved in the model \eqref{eq:Kuserapprox} for user-$1$ being the channel of interest. Since the coefficients are symmetric about $m=0$, only the positive time indices are shown. It is seen that for both curves, a channel memory of $M=11$ seems sufficient, since both the sets of coefficients are practically zero beyond $10$ symbols.
}
\label{fig:coeffs}
\end{figure}

\begin{figure}[!t]
\vspace{-6cm}
\begin{center}
\scalebox{1.0}{
\input{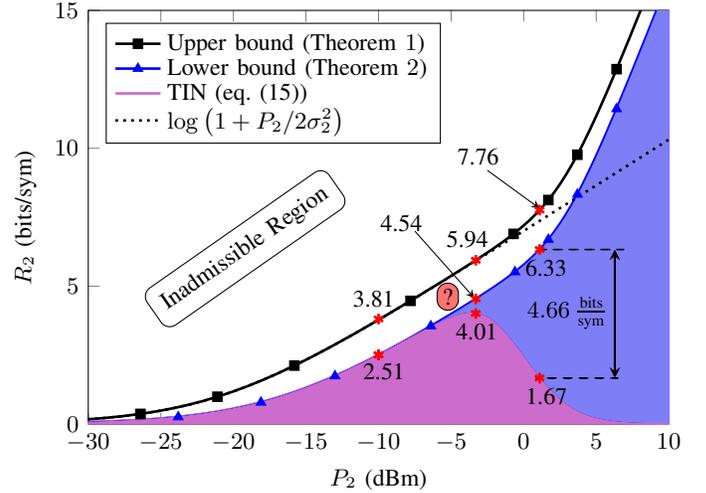}
}
\end{center}
\caption{ Upper bound in Theorem~\ref{thm:OBgenieK}, lower bound in Theorem~\ref{thm:ach} and the baseline scheme of treating interference as noise versus peak input power. The points marked in red correspond to peak input powers of $-10$ dBm, $-3.3$ dBm, and $1.1$ dBm, with the values of the $y-$coordinates on the corresponding capacity curves marked alongside. These values will be used in the three dimensional depiction of the corresponding rate regions in Fig.~\ref{fig:OBregion}. The black dotted curve depicts a $\log(1+\textrm{SNR})$ bound for user-$2$ that would be obtained if the SNR is taken to be simply $P_2/2\sigma_2^2$ (see corresponding discussion in Sec.~\ref{sec:num}).}
\label{fig:OBsnr}
\end{figure}

\begin{figure*}[!t]
\input{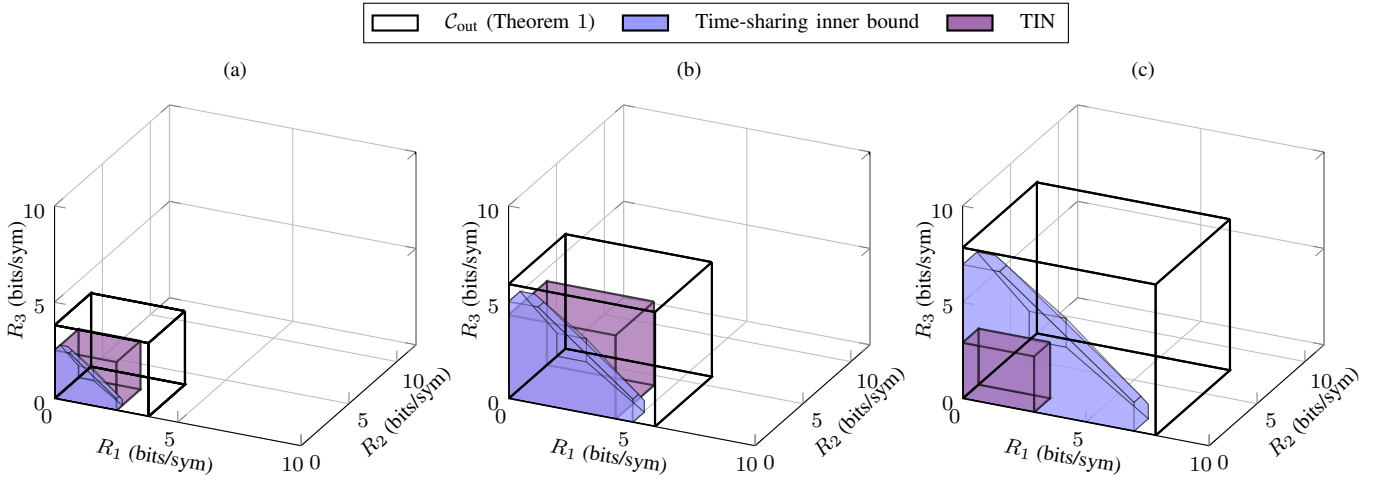}
\caption{Outer bound $\mathcal{C}_{\textup{out}}$ in Theorem~\ref{thm:OBgenieK} and inner bounds for a peak transmitted power per user of (a) $-10$ dBm, (b) $-3.3$ dBm, and (c) $1.1$ dBm. The intercepts of the different regions on each of the axes correspond to the red starred values in Fig.~\ref{fig:OBsnr}.}
\label{fig:OBregion}
\end{figure*}

\begin{figure}{}
\vspace{-5mm}
\begin{center}
\scalebox{0.65}{\input{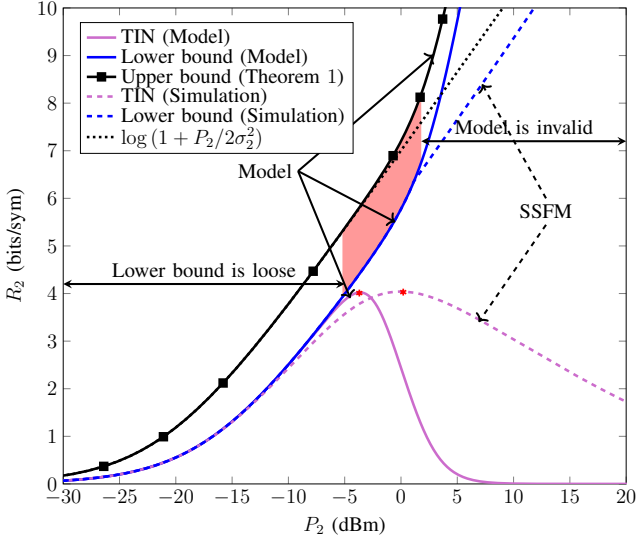}}
\end{center}
\caption{ Capacity lower bounds based on the model \eqref{eq:Kuserapprox} (solid lines) versus SSFM simulations (dashed lines). The red shaded area represents the regime where the proposed lower bound is useful and interesting (beyond the low-power regime). At higher values of launch powers, the model in \eqref{eq:Kuserapprox} becomes inaccurate. At low powers, the lower bound can be improved to $\log\left(1+{P_2}/({2\sigma_2^2 e})\right)$ by setting the interferer transmissions $|x_w[i]|^2$ to be zero instead of $P_w$ as in \eqref{eq:constampl}.}
\label{fig:OBregionSSFM}
\end{figure}

\subsection{Numerical Results} \label{sec:num}

\begin{table}[t]
\caption{Model Parameters}
\centering
\footnotesize
\begin{tabular}{c c}
\toprule
Parameter & Value \\ 
\midrule
{Memory Length $M$ (eq. \eqref{eq:Mnotation})} & {11}\\
Number of WDM users $K$ & 3 \\
Distance $L$ & $250 \: \textrm{km}$ \\
Nonlinearity parameter $\gamma$ & $1.2 \: \textrm{W}^{-1} \textrm{km}^{-1}$ \\
Signalling Rate & $32 \: \textrm{Gbaud}$ \\
Fiber attenuation $\alpha$ & $0.2 \: \textrm{dB/km}$ \\
Group velocity dispersion $\beta_2$ & $-21.7 \: \textrm{ps\textsuperscript{2}/km}$ \\
RRC pulse-shaping roll-off & 0.1\\
{Amplifier noise figure} & {$3 \: \textrm{dB}$}\\
Channel spacing $\Omega$ & $100 \: \textrm{GHz}$\\
\bottomrule
\end{tabular}
\label{table:param} 
\end{table}
\normalsize

The parameters used in our numerical results are summarized in Table \ref{table:param}. As in Example~\ref{ex:1}, we consider the case of $K=3$ WDM channels. {We have taken $M=11$ in \eqref{eq:Kuserapprox} since this is a good approximation to the channel memory for a single-span system of length $L=250 \: \textrm{km}$ and the signalling rate under consideration. See Fig.~\ref{fig:coeffs} for an illustration, where the absolute value of channel coefficients $|c_{1,2}^m|,|c_{1,3}^m|$ involved in the model \eqref{eq:Kuserapprox} are depicted for user-$1$ being the channel of interest. Note that the blue curve always lies above the red curve, since channels that are closer to the channel of interest have a stronger contribution to the nonlinear interference. The channel coefficients decrease in magnitude with $m$, and both sets of coefficients are seen to be practically zero beyond $10$ symbols. Hence, a value of $M=11$ captures the memory involved in the coupling between both pairs of channels}.

The upper bound on rates admissible for user$-2$ ($U_2$) in Theorem~\ref{thm:OBgenieK} and the corresponding lower bound in Theorem~\ref{thm:ach} are plotted in Fig. \ref{fig:OBsnr} against the peak input power for the symmetric case of $P_1=P_2=P_3$. Theorem~\ref{thm:ach} implies that any rate below $L_k$ in \eqref{eq:lbepi} is achievable, which gives the blue shaded area. The upper bound in \eqref{eq:ob1} from Theorem~\ref{thm:OBgenieK} gives an inadmissible region, which is not achievable. For comparison, we also plot the TIN bound in \eqref{eq:lbTIN} obtained by treating the interference term in \eqref{NLI.TIN} as Gaussian noise (whose variance is computed numerically), by choosing the inputs $X_k^n$ to be i.i.d. (with equal powers $P_k=P$ for all $k \in \{1,2,3\}$) according to the probability distribution $p_{X}(x)$, where the phase of $X$ is uniform on the interval $[-\pi,\pi]$ and independent of its amplitude $|X|=R$ that has probability density function given in \eqref{pdf.R}.
The choice of this distribution is motivated by commonly used achievability schemes for complex Gaussian channels with peak power constraints on the input~\cite{shamai1995capacity}. The resulting TIN region is depicted by the shaded purple area in Fig. \ref{fig:OBsnr}. The achievability of the area marked ``?" remains unknown. {We have also depicted (in dotted black) a $\log(1+\textrm{SNR})$ bound for the channel of interest that would be obtained if the SNR is taken to be simply $P_2/2\sigma_2^2$, without taking into account the factor $\left(1+\jmath \sum_{m\in\mathcal{M}} c_{2,w}^m \sum_{w\in\mathcal{W}_2}  P_w \right)$ that multiplies the input in the simplified model as a conse-} {quence of Lemma~\ref{lem:optinterf}. When the correct definition of SNR is applied, given by $P_2/2\sigma_2^2 \left(1+\left(\sum_{w\in\mathcal{W}_2}  P_w\sum_{m\in\mathcal{M}} c_{2,w}^m \right)^2\right)$, the resulting $\log(1+\textrm{SNR})$ bound coincides with the upper bound in Theorem~\ref{thm:OBgenieK}. The mismatch between these two bounds is due to the fact that the regular perturbative model in \eqref{eq:Kuser} is not} {energy preserving, as we alluded to earlier in Sec.~\ref{sec:chnmodel}.}

In Figs. \ref{fig:OBregion}(a)--(c), we plot the trade-off between the rates of the three users for fixed (and equal) powers of $-10$ dBm, $-3.3$ dBm and $1.1$ dBm, respectively (see the red stars in Fig.~\ref{fig:OBsnr}). The cuboidal region implied by the genie-aided outer bound in Theorem~\ref{thm:OBgenieK} is shown by solid black lines in Figs. \ref{fig:OBregion}(a)--(c). For comparison, we have depicted the respective achievable rate regions obtained by treating the interference terms in \eqref{eq:Kuserapprox} as Gaussian noise as the cuboidal regions in purple. Note that these \emph{interference as noise} regions eventually vanish in the highly nonlinear regime.

The strategy of constant-amplitude signaling for the interferers (users $w\in\mathcal{W}_k$) along with Theorem~\ref{thm:ach} for user-$k$, $k \in \{1,2,3\}$, defines 3 achievable rate triples on the 3-dimensional plane. Time-sharing between these achievable rate triples yields another inner bound for the channel in \eqref{eq:Kuserapprox}. {We further choose $16-$PSK alphabets for the interferers in Figs. \ref{fig:OBregion}(a)--(c)\footnote{We note here that if all the interfering users $w \in \mathcal{W}_k$ were to send a fixed sequence $P_w$ as in Lemma~\ref{lem:optinterf}, then their rates would be zero. But this need not necessarily be the case. The interfering users could achieve a non-zero rate by using phase-shift keying alphabets for instance, as depicted in Fig.~\ref{fig:OBregion}.}. When $k=2$ is the channel of interest, the interferer signals on channels $k=1$ and $k=3$ are chosen} { from $16-$PSK alphabets. The channel model for the user $k=1$ in this case becomes $Y_1[i]=X_1[i]+\sum_{m}\jmath c_{1,2}^m|X_2[i-{m]|^2 X_1[i]+\sum_{m}\jmath c_{1,3}^m|X_3[i-m]|^2 X_1[i]+N_1[i]}$, {wherein one} { of the interfering terms results from the channel of interest $k=2$ (multiplied by the PSK signal from $k=1$) while the other one comes from a PSK constellation for $k=3$ (multiplied by the PSK signal from $k=1$). The achievable rate for the user $k=1$ is now obtained by treating both kinds of interference terms as noise.}} These polyhedral regions corresponding to time-sharing are depicted in blue in Figs. \ref{fig:OBregion}(a)--(c). It is observed that this simple strategy of time-sharing between single-user codes outperforms treating interference as noise. 
In fact, the relative gains of time-sharing (in terms of better achievable rate tuples) compared to treating interference as noise becomes more pronounced with increasing powers. We note that these conclusions are valid as long as the channel model under consideration is valid. It is well accepted that the perturbative model under consideration is accurate for powers a few dB beyond the optimum launch power (which is $-3.3$ dBm in Fig.~\ref{fig:OBsnr}).

\begin{figure}
\begin{center}
\scalebox{0.6}{\begin{tikzpicture}[thick]
\node (e1) at (-2.9,1.3) [rectangle, draw, minimum height=1.0cm, minimum width=1.0cm]{${E_1}$};
\node (e11) at (-2.9,0) [rectangle, draw, minimum height=1.0cm, minimum width=1.0cm]{${E_2}$};
\node (e2) at (-2.9,-1.4) [rectangle, draw, minimum height=1.0cm, minimum width=1.0cm]{${E_3}$};
\node[rectangle, draw, minimum height=2.0cm, minimum width=3.5cm, text width=2.75cm,rounded corners,align=center] (t1) at (2.1,0) {{Tx processing $+$ \textbf{NLSE channel \eqref{eq:nlse}} $+$ Rx processing}};
\node[rectangle, draw, minimum height=2.0cm, minimum width=3.5cm, text width=3.5cm,rounded corners,align=center] (r1) at (8,0) {{Compute the rate $\frac{1}{n}\sum_{i=1}^n I(X_k[i];Y_k[i])$ for $k=1,2,3$}};
\draw[->] (e1) --++(t1) node[midway, sloped, above]{$X_1^n \sim p_{X_1^n}$};
\draw[->] (e11) --++(t1) node[midway, above]{$X_2^n \sim p_{X_2^n}$};
\draw[->] (e2) --++(t1) node[midway, sloped, above]{$X_3^n \sim p_{X_3^n}$};
\draw[->] (t1.25) -- (r1.166) node[midway, above]{ $Y_1^n$};
\draw[->] (t1.0) -- (r1.180) node[midway, above]{$Y_2^n$};
\draw[->] (t1.-23) -- (r1.193) node[midway, above]{$Y_3^n$};
\end{tikzpicture}}
\caption{Schematic for computation of the bounds via SSFM simulations. The $k$-th user transmits a signal  $X_k^n$ according to the distribution $p_{X_k^n}$ using an encoder $E_k$, for $k=1,2,3$. The input distributions are shown in Table~\ref{table:pdf} for computation of the TIN rate as well as the lower bound (analogous to Theorem~\ref{thm:ach}). After propagation over the NLSE channel described by \eqref{eq:nlse}, the outputs $Y_k^n, k=1,2,3$ are used to compute the achievable rate $\frac{1}{n}\sum_{i=1}^n I(X_k[i];Y_k[i])$ for $k=1,2,3$.}
\label{fig:ssfm}
\end{center}
\end{figure}
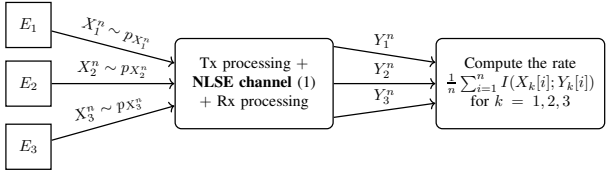

\begin{table}[t]
\caption{Input distributions for SSFM simulations (for $k=2$ being the channel of interest)}
\centering
\footnotesize
\begin{tabular}{c c}
\toprule
{Scenario} & {Distribution} \\ 
\midrule
 & {$p_{X_1^n}$: i.i.d. with $|X_1|=R_1 \sim p_{R_1}(r)$ from \eqref{pdf.R}}\\
{TIN bound} & {$p_{X_2^n}$: i.i.d. with $|X_2|=R_2 \sim p_{R_2}(r)$ from \eqref{pdf.R}}\\
 & {$p_{X_3^n}$: i.i.d. with $|X_3|=R_3 \sim p_{R_3}(r)$ from \eqref{pdf.R}}\\
\midrule
 & {$p_{X_1^n}$: i.i.d. from a $16$-PSK constellation}\\
{Lower bound (SSFM)} & {$p_{X_2^n}$: i.i.d. with $|X_2|=R_2 \sim p_{R_2}(r)$ from \eqref{pdf.R}}\\
 & {$p_{X_3^n}$: i.i.d. from a $16$-PSK constellation}\\
\bottomrule
\end{tabular}
\label{table:pdf} 
\end{table}
\normalsize

{In Fig.~\ref{fig:OBregionSSFM}, we compare the capacity bounds in Theorem~\ref{thm:ach}} { and TIN obtained from the model \eqref{eq:Kuserapprox} with split-step Fourier method (SSFM) simulations. The simulation setup to compute the achievable rates is depicted schematically in Fig.~\ref{fig:ssfm}. The $k$-th user transmits a signal  $X_k^n$ according to the distribution $p_{X_k^n}$ for $k=1,2,3$. The transmitted symbols are either chosen i.i.d. according to the probability distribution in \eqref{pdf.R} or taken i.i.d. from a $16-$PSK constellation -- see Table~\ref{table:pdf} for the details corresponding to the different scenarios. After propagation over the NLSE channel described by \eqref{eq:nlse}, the outputs $Y_k^n, k=1,2,3$ are realized. We then compute the} { achievable rate $\frac{1}{n}\sum_{i=1}^n I(X_k[i];Y_k[i])$ for $k=1,2,3$. The SSFM calculations are performed using uniform spacial step sizes of 0.01 km.}

{It is observed that the power at which the TIN curves peak differs between the model and simulations, albeit the same peak values are attained using both approaches. Furthermore, it is observed that beyond a certain power (around $2$ dBm in Fig.~\ref{fig:OBregionSSFM}), the capacity lower bound from simulations is smaller than that predicted by the model. This could possibly be attributed to the inaccuracy of the adopted simplified model (given by \eqref{eq:Kuserapprox}) at higher transmit powers. On the other hand, at low powers, the lower bound in Theorem~\ref{thm:ach} can be improved to $\log\left(1+{P_2}/({2\sigma_2^2 e})\right)$ by setting the interferer transmissions $|x_w[i]|^2$ to be zero instead of $P_w$ as in \eqref{eq:constampl}. Therefore, in Fig.~\ref{fig:OBregionSSFM}, the shaded area in red represents the regime where the proposed lower bound is useful and interesting (beyond the low-power regime).}

\section{Conclusions} \label{sec:conc}

We took a multi-user information theoretic view of a $K$-user wavelength division multiplexing system impaired by cross-phase modulation and dispersion, and derived a novel capacity region outer bound using genie-aided techniques. An achievable rate region was also obtained for the same, and it was shown that time-sharing between certain single-user schemes can strictly outperform treating interference as noise. Though we assumed that SPM is ideally compensated in our model, we believe that the results in this paper can be generalized to take into account both SPM as well as XPM. 

This paper is a very first step towards a multi-user characterization of fiber optic systems with realistic channel models, breaking away from the traditional single-user perspective. Future works include obtaining tighter achievable regions/inner bounds as well as outer bounds, and the design and implementation of schemes that can achieve the presented capacity bounds in practice. {Moreover, an extension of the current results to a multi-span situation with signal-noise interactions seems to be an interesting avenue for further research.}

\section*{Acknowledgements}
The authors would like to thank Dr. Hamdi Joudeh (Eindhoven University of Technology) for fruitful discussions on the channel model and the proofs of Theorems~\ref{thm:OBgenieK} and \ref{thm:ach}. {The authors are also grateful to the Editor and the two anonymous reviewers whose insightful comments have greatly improved the quality of the paper and the exposition.}

\appendices
\section{Proof of lemma~\ref{lem:optinterf}} \label{app:lem:genieK}
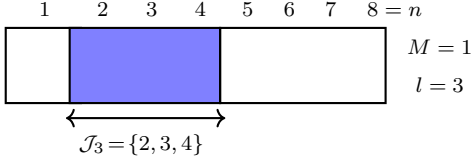
\begin{figure}
\begin{center}
\begin{tikzpicture}[thick]
\node (a4) at (0,0) [rectangle, draw, minimum width=1cm, minimum height=1cm]{};
\node (a5) at (1.35,0) [rectangle, fill=blue!50!white, draw, minimum width=2cm, minimum height=1cm]{};
\node (a6) at (3.45,0) [rectangle, draw, minimum width=2.2cm, minimum height=1cm]{};
\node () at (0.02,0.75) {\footnotesize $1$};
\node () at (0.79,0.75) {\footnotesize $2$};
\node () at (1.44,0.75) {\footnotesize $3$};
\node () at (2.09,0.75) {\footnotesize $4$};
\node () at (2.72,0.75) {\footnotesize $5$};
\node () at (3.27,0.75) {\footnotesize $6$};
\node () at (3.82,0.75) {\footnotesize $7$};
\node () at (4.67,0.75) {\footnotesize $8=n$};
\draw[<->] (0.30,-0.7) -- (2.42,-0.7);
\node () at (1.30,-1.05) {\footnotesize $\mathcal{J}_3\!=\!\{2,3,4\}$};
\node () at (5.27,0.25) {\footnotesize $M=1$};
\node () at (5.27,-0.25) {\footnotesize $l=3$};
\end{tikzpicture}
\end{center}
\caption{Illustration of the inequality \eqref{eq:lemineq1} for $M=1$ and $n=8$. The weighted sum on the left-hand side of \eqref{eq:lemineq1} is upper bounded by replacing the interferer squared amplitude terms by its maximum value over the interval $\mathcal{J}_3$.}
\label{fig:lem1}
\end{figure}
We first note that for any given $l$ such that $M \leq l \leq n-M$, the symbols $X_w[l]$ are well defined. We then have
\begin{align}
\sum_{m \in \mathcal{M}} {c_{k,w}^m} |x_w[l-m]|^2 \leq \max_{j \in \mathcal{J}_l} |x_w[j]|^2 \left(\sum_{m \in \mathcal{M}} {c_{k,w}^m}\right), \label{eq:lemineq1}
\end{align}
where the set $\mathcal{J}_l \triangleq \{l-M,l-M+1,\ldots,l+M\}$ is a set of indices for the sliding window depicted in Fig.~\ref{fig:lem1}. In other words, we upper bound each of the $|x_w[l-m]|^2$ terms in the weighted sum $\sum_{m \in \mathcal{M}} {c_{k,w}^m} |x_w[l-m]|^2$ by its maximum value over the interval $\mathcal{J}_l$. Furthermore,
\begin{align}
\max_{j \in \mathcal{J}_l} |x_w[j]|^2 &\leq \max_{j \in \{1,2,\ldots,n\}} |x_w[j]|^2 \notag\\
&\leq P_w, \label{eq:lemineq2}
\end{align}
where the last step follows from the peak power constraint on user $w$.
Applying the inequality \eqref{eq:lemineq2} in \eqref{eq:lemineq1}, we obtain
\begin{align} \label{eq:lemproof}
\sum_{m \in \mathcal{M}} {c_{k,w}^m} |x_w[l-m]|^2 \leq P_w \left(\sum_{m \in \mathcal{M}} {c_{k,w}^m}\right),
\end{align}
as desired.
Clearly, equality holds in \eqref{eq:lemproof} with the choice of
\begin{align}
|x_{w}[i-m]|^2 = P_w, \: \forall \:w\in\mathcal{W}_k, \: m \in \mathcal{M}, 1 \leq i-m \leq n. \label{eq:optinterf}
\end{align}
This completes the proof.

\section{Proof of Theorem \ref{thm:OBgenieK}} \label{app:genieK}
We now establish the outer bound using information theoretic inequalities. 
The rate of user$-k$, $k \in \{1,2,\ldots,K\}$, can be upper bounded as follows:
\begin{align}
&nR_k \notag\\
&{\stackrel{(a)}=} H(M_k) \notag\\
&\stackrel{(b)}= H(M_k|\{X_w^n | w\in\mathcal{W}_k\}) \notag\\
&= H(M_k|\{X_w^n | w\in\mathcal{W}_k\})-H(M_k|\{X_w^n | w\in\mathcal{W}_k\},Y_k^n) \notag\\
&\phantom{wwwwwwwwwwwwww}+H(M_k|\{X_w^n | w\in\mathcal{W}_k\},Y_k^n) \notag\\
&\stackrel{(c)}\leq I(M_k;Y_k^n|\{X_w^n | w\in\mathcal{W}_k\})+H(M_k|Y_k^n) \notag\\
&\stackrel{(d)}\leq I(M_k;Y_k^n|\{X_w^n | w\in\mathcal{W}_k\})+1+P_{e} nR_k  \notag\\
&\stackrel{(e)}= I(M_k;Y_k^n|\{X_w^n | w\in\mathcal{W}_k\})+n\epsilon_n  \notag\\
&\stackrel{(f)}\leq I(X_k^n;Y_k^n|\{X_w^n | w\in\mathcal{W}_k\})\!+\!n\epsilon_n \notag\\
&= h(Y_k^n|\{X_w^n | w\in\mathcal{W}_k\})-\!h(Y_k^n|X_k^n,\{X_w^n | w\in\mathcal{W}_k\})+\!n\epsilon_n \notag\\
&{\stackrel{(g)}=} h(Y_k^n|\{X_w^n | w\in\mathcal{W}_k\})-\!h(N_k^n)+\!n\epsilon_n \notag\\
&\stackrel{(h)}\leq \sum_{i=1}^{n} h(Y_{k}[i]|\{X_w^n | w\in\mathcal{W}_k\})-\sum_{i=1}^{n} h(N_{k}[i])+n\epsilon_n \notag\\
&\stackrel{(i)}\leq \sum_{i=1}^{n} \max_{\{x_{w}^n | w\in\mathcal{W}_k\}} \left[ h(Y_{k}[i]|\{X_w^n=x_{w}^n | w\in\mathcal{W}_k\})\right] \notag\\
&\phantom{wwwww}-\sum_{i=1}^{n} h(N_{k}[i])+n\epsilon_n \notag\\
&\stackrel{(j)}\leq  \sum_{i=1}^{n} \max_{\{x_{w}^n | w\in\mathcal{W}_k\}} \frac{1}{2}\!\log\!\left(\!{ \textup{det}\!\left(\!\textup{cov}\!\left(\begin{aligned}\!&Y_{k}^R[i],Y_{k}^I[i]\Big| \\ &\!\{X_w^n=x_{w}^n | w\in\mathcal{W}_k\} \end{aligned} \!\right)\!\right)}\!\right) \notag\\
&\phantom{wwwww}-n\log\left(\sigma_k^2\right)+n\epsilon_n, \label{eq:convKg}
\end{align}
{where (a) follows from the assumption that the messages are uniformly distributed over their respective ranges}, (b) follows since $M_k$ is independent of $\{X_w^n | w\in\mathcal{W}_k\}$ with the set $\mathcal{W}_k$ being defined in \eqref{eq:Wnotation}, (c) follows since conditioning does not increase the entropy, (d) follows from Fano's inequality with $P_e$ being defined as in \eqref{eq:errorprob}, (e) follows by defining $\epsilon_n =(1/n+P_{e} R_k)$ with $\epsilon_n \xrightarrow{n \to \infty} 0$, (f) follows from the data processing inequality since $M_k \to X_k^n \to Y_k^n$ form a Markov chain conditioned on $\{X_w^n | w\in\mathcal{W}_k\}$, {(g) follows since conditioned on all the inputs, the only remaining uncertainty in $Y_k^n$ is due to $N_k^n$}, (h) follows since conditioning does not increase the entropy and the fact that the additive noise is i.i.d., (i) follows since $h(Y_{k}[i]|\{X_w^n | w\in\mathcal{W}_k\})$ represents an average over $\{x_w^n | w\in\mathcal{W}_k\}$ and the average is upper bounded by the maximum, while (j) follows from the fact that Gaussian random vectors maximize the differential entropy under a covariance constraint. We note that the $\max$ over $x_w^n$ sequences in steps (g) and (h) are subject to the peak power constraint $\underset{i \in \{1,\ldots,n\}}{\max} |x_{w}[i]|^2 \leq P_w \: \forall w \in \mathcal{W}_k$.

It now remains to bound the $\log(\textup{det}(\cdot))$ terms in expression \eqref{eq:convKg}. On expressing equation \eqref{eq:Kuserapprox} in terms of its respective real and imaginary components, we have: 
\begin{align}
Y_{k}^R[i] &\!=\! X_{k}^R[i]\!-\!\sum_{\substack{w\in\mathcal{W}_k}}  \sum_{m\in\mathcal{M}} {c_{k,w}^m} |X_{w}[i-m]|^2 X_{k}^I[i] \!+\!N_{k}^R[i], \label{eq:Kvar1} \\
Y_{k}^I[i] &\!=\! X_{k}^I[i]\!+\!\sum_{\substack{w\in\mathcal{W}_k}}  \sum_{m\in\mathcal{M}} {c_{k,w}^m} |X_{w}[i-m]|^2 X_{k}^R[i]\!+\!N_{k}^I[i]. \label{eq:Kvar2}
\end{align}
Let $\mathbb{E}[(X_{k}^R[i])^2]=p_{k,i}^R$ and $\mathbb{E}[(X_{k}^I[i])^2]=p_{k,i}^I$ be the instantaneous powers associated with the real and imaginary parts of $X_k[i]$. Since the sum of these powers constitute the instantaneous power of $X_k[i]$, we write
\begin{align}
p_{k,i}^R+p_{k,i}^I \leq P_{k,i}, \label{eq:powerinst}
\end{align}
where $\mathbb{E}[|X_k[i]|^2] \leq P_{k,i}$ from the power constraint.
Hence we can write the following chain of inequalities for the determinant of the covariance matrix involved in \eqref{eq:convKg}:
\begin{align}
&\textup{det}(\textup{cov}(Y_{k}^R[i],Y_{k}^I[i]|\{X_w^n=x_{w}^n | w\in\mathcal{W}_k\})) \notag\\
&=\textup{det}\!\Bigg(\!\textup{cov}\!\Bigg(\!X_{k}^R[i] \notag\\
&\phantom{wwwwwww}-\sum_{\substack{w\in\mathcal{W}_k}}  \sum_{m\in\mathcal{M}} {c_{k,w}^m} |X_{w}[i-m]|^2 X_{k}^I[i] +N_{k}^R[i], \notag\\
&\phantom{wwwwwww} X_{k}^I[i] \notag\\
&\phantom{wwwwwww}+\sum_{\substack{w\in\mathcal{W}_k}}  \sum_{m\in\mathcal{M}} {c_{k,w}^m} |X_{w}[i-m]|^2 X_{k}^R[i]+N_{k}^I[i]  \notag\\
&\phantom{wwwwwwwwwwwwwwwwww} \Bigg|\{X_w^n=x_{w}^n | w\in\mathcal{W}_k\}\Bigg)\!\!\Bigg) \notag\\
&\stackrel{(a)}=\textup{det}\!\Bigg(\!\textup{cov}\!\Bigg(\!X_{k}^R[i] \notag\\
&\phantom{wwww}-\sum_{\substack{w\in\mathcal{W}_k}}  \sum_{m\in\mathcal{M}} {c_{k,w}^m} |x_{w}[i-m]|^2 X_{k}^I[i] +N_{k}^R[i], \notag\\
&\phantom{wwwwww} X_{k}^I[i] \notag\\
&\phantom{wwwww}+\sum_{\substack{w\in\mathcal{W}_k}}  \sum_{m\in\mathcal{M}} {c_{k,w}^m} |x_{w}[i-m]|^2 X_{k}^R[i]+N_{k}^I[i]\Bigg)\!\!\Bigg) \notag\\
&\stackrel{(b)}\leq \!\frac{1}{4}\!\left(\!\begin{aligned}&\textup{var}\Bigg(\!X_{k}^R[i] \notag\\
&\phantom{w}-\sum_{\substack{w\in\mathcal{W}_k}}  \sum_{m\in\mathcal{M}} {c_{k,w}^m} |x_{w}[i-m]|^2 X_{k}^I[i] +N_{k}^R[i] \!\Bigg)\\ &\!\!\!\!+\!\textup{var}\Bigg(\!X_{k}^I[i] \notag\\
&\phantom{w}+\sum_{\substack{w\in\mathcal{W}_k}}  \sum_{m\in\mathcal{M}} {c_{k,w}^m} |x_{w}[i-m]|^2 X_{k}^R[i]+N_{k}^I[i] \!\Bigg)\end{aligned}\!\!\right)^2 \notag\\
&= \frac{1}{4}\Bigg(p_{k,i}^R\!+\!p_{k,i}^I\!+\!2\sigma_k^2  \notag\\
&\phantom{wwww} +\left(\sum_{\substack{w\in\mathcal{W}_k}}  \sum_{m\in\mathcal{M}} {c_{k,w}^m} |x_{w}[i-m]|^2\right)^2\!\! (p_{k,i}^R+p_{k,i}^I)\Bigg)^2 \notag\\
&\stackrel{(c)}\leq \frac{1}{4}\left(\!P_{k,i}\!\left(\!1\!+\!\left(\!\sum_{\substack{w\in\mathcal{W}_k}}  \sum_{m\in\mathcal{M}} {c_{k,w}^m} |x_{w}[i-m]|^2 \!\right)^2 \!\right)\!+\!2\sigma_k^2 \!\right)^2\!, \label{eq:Kuser2}
\end{align}
where (a) follows from the independence of the inputs $X_k^n$ and $X_w^n, w \in \mathcal{W}_k$, (b) follows since $\textup{det}(A) \leq \left(\frac{\textup{trace}(A)}{n}\right)^n$ for any $n\times n$ square matrix $A$, while (c) follows from \eqref{eq:powerinst}.
From expressions \eqref{eq:convKg} and \eqref{eq:Kuser2}, we obtain the following expression for an upper bound on the rate achievable by user-$k$, $k \in \{1,2,\ldots,K\}$:
\begin{align}
&n(R_k-\epsilon_n) \notag\\
&\leq \sum_{i=1}^{n} \max_{\{x_{w}^n | w\in\mathcal{W}_k\}} \notag\\
&\phantom{www}\!\left[\!\log\!\!\left(\!\!1\!+\!\frac{P_{k,i}}{2\sigma_k^2}\!\!\left(\!1\!+\!\left(\sum_{\substack{w\in\mathcal{W}_k}}  \sum_{m\in\mathcal{M}} {c_{k,w}^m} |x_{w}[i-m]|^2 \!\right)^2 \!\right)\!\right)\!\right]\!\!. \label{eq:rate1K}
\end{align}
Expression \eqref{eq:rate1K} involves finding the optimal interferer realizations $\{x_{w}^n | w\in\mathcal{W}_k\}$ that maximize the rate $R_k$.
The objective function in \eqref{eq:rate1K} is of the form $\log(1+c(1+f(x)^2))$ with $c$ being a constant, and the function $f(\cdot)$ involves \smash{$\sum_{m \in \mathcal{M}} {c_{k,w}^m} |x_{w}[i-m]|^2$}. This function only takes on nonnegative values due to \eqref{c.positive}, and $|x_{w}[i-m]|^2$ is nonnegative as well. Therefore, by monotonicity, it suffices to maximize $f(x)$ instead of $\log(1+c(1+f(x)^2))$, and we are interested in the following optimization problem:
\begin{align}
O_i=\sum_{\substack{w\in\mathcal{W}_k}}  \: \max_{\substack{x_{w}^n : \\ \underset{m \in \{1,\ldots,n\}}{\max} |x_{w}[m]|^2 \leq P_w}} \: \sum_{m\in\mathcal{M}} {c_{k,w}^m} |x_{w}[i-m]|^2. \label{eq:opt1K}
\end{align}
In particular, this allows us to write (from \eqref{eq:rate1K})
\begin{align}
n(R_k-\epsilon_n) \leq \sum_{i=1}^{n} \log\left(1+\frac{P_{k,i}}{2\sigma_k^2}\left(1+O_i^2 \right)\right). \label{eq:optimK}
\end{align}
Applying the inequality \eqref{eq:lem1ineq} in \eqref{eq:opt1K}, we obtain
\begin{align}
O_i \leq \sum_{\substack{w\in\mathcal{W}_k}} P_w \left(\sum_{m\in\mathcal{M}} {c_{k,w}^m}\right). \label{eq:holder4K}
\end{align}
Using \eqref{eq:holder4K}, the upper bound in \eqref{eq:optimK} now becomes
\begin{align}
&n(R_k-\epsilon_n) \notag\\
&\leq \sum_{i=1}^{n} \!\left[\!\log\!\left(\!1\!+\!\frac{P_{k,i}}{2\sigma_k^2}\!\left(1+\left(\sum_{\substack{w\in\mathcal{W}_k}}  P_w \sum_{m\in\mathcal{M}} {c_{k,w}^m}\right)^2\right)\!\right)\!\right] \notag\\
&\stackrel{(a)}\leq n\log\!\left(\!1\!+\!\frac{1}{n}\!\sum_{i=1}^{n}\! \frac{P_{k,i}}{2\sigma_k^2}\!\left(\!1\!+\!\left(\!\sum_{\substack{w\in\mathcal{W}_k}}  P_w \sum_{m\in\mathcal{M}} {c_{k,w}^m}\right)^2 \right)  \!\right) \notag\\
&\stackrel{(b)}\leq n\log\left(1+\frac{P_{k}}{2\sigma_k^2}\left(1+\left(\sum_{\substack{w\in\mathcal{W}_k}}  P_w \sum_{m\in\mathcal{M}} {c_{k,w}^m}\right)^2\right) \right), \label{eq:convf1K}
\end{align}
where (a) follows from Jensen's inequality, while (b) follows since the maximum power constraint implies the average power constraint $\sum_{i=1}^n P_{k,i} \leq nP_k$.
Dividing throughout by $n$ and letting $n \to \infty$ (which makes $\epsilon_n \to 0$) completes the proof of the upper bound in Theorem~\ref{thm:OBgenieK}.

\section{Proof of Theorem~\ref{thm:ach}} \label{sec:appach}
Consider the memoryless single-user channel in \eqref{eq:shamai}.
\begin{align}
Y_{k} &= X_{k}\left(1+\jmath \sum_{m\in\mathcal{M}} {c_{k,w}^m} \sum_{\substack{w\in\mathcal{W}_k}}  P_w \right)+N_{k}.
\end{align}
The mutual information between $X_k$ and $Y_k$ can be bounded as:
\begin{align}
&I(X_k;Y_k) \notag\\
&= h(Y_k)-h(Y_k|X_k) \notag\\
&= h(Y_k)-h(N_k) \notag\\
&\stackrel{(a)} \geq \ln\left(e^{h\left(X_{k}\left(1+\jmath \sum_{\substack{w\in\mathcal{W}_k}}  P_w \sum_{m\in\mathcal{M}} {c_{k,w}^m} \right)\right)}+e^{h(N_k)}\right) \notag\\
&\phantom{wwwwww}-\ln(2 \pi e \sigma_k^2) \notag\\
&{\stackrel{(b)}=} \ln\!\left(\!e^{h(X_k)+\ln\left(1+\left(\sum_{\substack{w\in\mathcal{W}_k}}  P_w \sum_{m\in\mathcal{M}} {c_{k,w}^m}\right)^2\right)}\!+\!2 \pi e \sigma_k^2 \!\right) \notag\\
&\phantom{wwwwww}-\ln(2 \pi e \sigma_k^2), \label{eq:ach1}
\end{align}
where (a) follows from the entropy power inequality, {while (b) follows from the scaling property of differential entropy}.
Now we choose the input distribution of $X_k$ as in \cite[eq.~(30)]{shamai1995capacity} to maximize the differential entropy $h(X_k)$, with the phase of $X_k$ being uniform on $[-\pi,\pi]$ and independent of the amplitude $|X_k|=R_k$ that has the probability density function given in \eqref{pdf.R}.
This leads to~\cite[eq.~(37)]{shamai1995capacity}
\begin{align}
h(X_k) &= \ln(\pi P_k). \label{eq:ach2}
\end{align}
Substituting \eqref{eq:ach2} in \eqref{eq:ach1}, we obtain
\begin{align}
&I(X_k;Y_k) \notag\\
&\geq \ln\!\left(\!e^{\ln(\pi P_k)+\ln\left(1+\left(\sum_{\substack{w\in\mathcal{W}_k}}  P_w \sum_{m\in\mathcal{M}} {c_{k,w}^m}\right)^2\right)}\!+\!2 \pi e \sigma_k^2 \!\right) \notag\\
&\phantom{wwwwww}-\ln(2 \pi e \sigma_k^2) \notag\\
&= \ln\left(1+\frac{P_{k}}{2\sigma_k^2 e}\left(1+\left(\sum_{\substack{w\in\mathcal{W}_k}}  P_w \sum_{m\in\mathcal{M}} {c_{k,w}^m}\right)^2\right) \right) \: \textup{nats} \notag\\
&= \log\left(1+\frac{P_{k}}{2\sigma_k^2 e}\left(1+\left(\sum_{\substack{w\in\mathcal{W}_k}}  P_w \sum_{m\in\mathcal{M}} {c_{k,w}^m}\right)^2\right) \right) \: \textup{bits}.
\end{align}
This completes the proof.

\nocite{}
\bibliographystyle{IEEEtran}
\bibliography{mylitoptics.bib}

\end{document}